\documentclass[preprint,10pt,3p]{elsarticle}

\usepackage{amssymb}
\usepackage{amsthm}
\usepackage{amsmath}
\usepackage{subcaption}
\usepackage{comment}
\usepackage[ruled]{algorithm2e}
\usepackage{booktabs}
\usepackage{colortbl}
\usepackage{graphicx}
\usepackage{xcolor}

\newtheorem{theorem}{Theorem} 
\newtheorem{proposition}{Proposition} 

\newtheorem{corollary}{Corollary}[theorem]

\newtheorem{remark}{Remark}

\AtBeginEnvironment{proof}{\footnotesize}

\journal{}

\begin{document}
	
	\begin{frontmatter}
		
		
		
		\title{Urn Modeling of Random Graphs Across Granularity Scales: A Framework for Origin-Destination Human Mobility Networks}
		
		
		\author[inst1]{Fabio Vanni}
		\affiliation[inst1]{organization={Department of Economics, University of Insubria},
			addressline={Via Monte Generoso, 71}, 
			city={Varese},
			postcode={21100}, 
			country={ITALY}}
		
		\author[inst2]{David Lambert}
		
		\affiliation[inst2]{organization={Department of Mathematics, University of North Texas},
			addressline={225 S. Avenue B}, 
			city={Denton},
			postcode={76201}, 
			state={Texas},
			country={USA}}

\begin{abstract}
	We model human mobility as a combinatorial allocation process, treating trips as distinguishable balls assigned to location-bins and generating origin–destination (OD) networks. From this analogy, we construct a unified three-scale framework, enumerative, probabilistic, and continuum graphon ensembles, and prove a renormalization theorem showing that, in the large sparse regime, these representations converge to a universal mixed-Poisson law. The framework yields compact formulas for key mobility observables, including destination occupancy, vacancy of unvisited sites, coverage (a stopping-time extension of the coupon collector problem), and overflow beyond finite capacities. Simulations with gravity-like kernels, calibrated on empirical OD data, closely match the asymptotic predictions. By connecting exact combinatorial models with continuum analysis, the results offer a principled toolkit for synthetic network generation, congestion assessment, and the design of sustainable urban mobility policies.
\end{abstract}

		
	\end{frontmatter}
	
	
\section{Introduction}
	Human mobility, the movement of people and goods across space, shapes urban form, economic activity, and social interaction. Origin-destination (OD) models are a core tool for analyzing these flows, linking origins (e.g., homes, firms) to destinations (e.g., workplaces, markets) through directed connections that represent commuting, migration, trade, and related interactions \cite{barthelemy2016structure,bettencourt2021introduction,du2024unveiling}. Rooted in geography, economics, statistical physics, and applied mathematics, OD systems are naturally represented as directed graphs, where nodes are administrative units and edges quantify flow intensities \cite{barthelemy2011spatial,barbosa2018human,aoki2022urban}, integrating spatial interaction and location theory \cite{chan2001location,de2024modelling}.

	We model OD networks via a formal analogy with combinatorial allocation problems, specifically balls-into-bins, in the spirit of \cite{casiraghi2021configuration,armenter2014balls}. In this view,  the OD mobility network is characterized by directed flows of individuals between locations \cite{nie2021understanding,vanni2024visit}, represented as nodes. The flows (balls) have origins (types or colors of balls) and destinations (bins), creating a complex, directed network of mobility patterns.  By leveraging this analogy, we can derive insights into the distribution of flows, congestion effects, and optimization strategies. 
	In particular, we introduce a three-level granularity framework for random graphs, based on a systematic coarsening process that combines statistical relaxation \cite{cimini2019statistical,bianconi2009entropy,farkas2004equilibrium} and structural aggregation \cite{hartle2021dynamic,bollobas2007phase}.
	Statistical relaxation replaces exact link allocations (fine-grained representation) with probabilistic ones (coarse-grain representation), moving from enumerative graph configurations to random ensembles that match expected link patterns. Structural aggregation  groups nodes according to latent characteristics (namely attractiveness such as location or socio-economic attributes), reducing the network’s dimensionality while preserving key structural heterogeneity by using the inhomogeneous random graph configurations  which  is a class of {random graphs with latent variables}, reformulated using urn modeling methods and stochastic occupancy equations in a continuous setting (continuum-grained representation). This procedure defines three renormalized ensembles,  each representing the system at a different resolution.

	As granularity decreases, these ensembles preserve the key large-scale statistics of the original network while filtering microscopic variability, revealing persistent topological structure across scales.
	At the fine-grained level, the system is described by a fully specified link allocation, where the randomness reflects combinatorial uncertainty over admissible configurations under exact structural constraints.  At the fine-grained level, we solve occupancy and load problems using exact combinatorial enumeration, refined with asymptotic approximations. In contrast, the coarse-grained representation relaxes these constraints and models the network probabilistically, interpreting connection tendencies as the result of partial information or latent structural preferences. At the coarse-grained level, we reformulate the problem in terms of cumulants and apply the saddle-point approximation to the moment-generating function. While the two approaches differ in their treatment of uncertainty, they become statistically consistent under appropriate large-scale and sparse conditions via a mean-field argument (self-averaging limit). Both frameworks yield a discrete representation of the system, making them well-suited for generating synthetic networks and producing numerical Monte Carlo results. 
    On the other side, at the continuum-grained level, the system is described through a class of random graphs with latent variables, where the discrete set of nodes is replaced by a continuous domain endowed with a probability density. Each node is characterized by a latent variable, and the connection or allocation process is governed by a kernel function acting on these latent traits.  We evaluate macroscopic observables with integral transforms and the Laplace method. Here, the large-deviation principle for the graphon (adjacency-matrix density) provides the variational rate functional that makes the continuum approach fully consistent with the discrete ones. The principal motivation for introducing this continuum formulation is to provide a general framework in which analytical expressions for typical urn problems, such as occupancy and other load-allocation problems, can be rigorously derived. 
    In this setting, key probabilistic quantities admit tractable integral representations in the large-system limit, enabling closed-form approximations that complement and extend the discrete Monte Carlo simulations produced by the fine- and coarse-grained ensembles.
    
	Section~\ref{sec_ensemblegraph} develops the mathematical foundations of these three urn-based ensembles. We formalize the fine-, coarse-, and continuum-grained representations of random pseudo-graphs and prove their asymptotic equivalence in the large, sparse regime (Theorems~\ref{th_macromicro} and \ref{theo_graphon}). 	
	Under these conditions, the coarse-grained conditioned product of binomials ensemble, the fine-grained multivariate hypergeometric  ensemble, and the exponential graphon (under fixed mass) describe the same limiting distribution.
	
	In Section~\ref{sec_occupancy}, we derive analytical expressions for the occupancy problem at all three scales. In mobility terms, this is the allocation of visits across locations; in graph terms, it is the OD network degree distribution. Specifically, it describes the expected number of location-bins containing a given number of visit-balls. We will show that, across all three granularity scales, this problem is asymptotically equivalent to a pure multinomial balls-into-bins setting. Furthermore, under suitable mean-field assumptions, the degree (visit) distribution converges to a mixed Poisson law, extending classical Poissonization results \cite{mitzenmacher2017probability,sevast1973poisson}.

	Section~\ref{sec_load} addresses allocation-load phenomena, i.e. vacancy, coverage, and overflow, consistently across the three scales. Vacancy counts locations left unvisited after a given number of trips, indicating underutilization and access inequities. Coverage adopts a stopping-time perspective, the number of trips required until exactly a target number of locations remain unvisited, generalizing the coupon collector problem. Last, overflow quantifies allocations exceeding finite capacities, providing an indicator of congestion and stress at destinations (e.g., stations, hospitals, roads). These metrics offer a compact, tractable basis for diagnosing crowding, queuing, and systemic strain, and for evaluating policies such as demand redistribution or capacity expansion.

	Finally, in Section~\ref{sec_simul}, we validate our results using parameterization and functional forms from the data-driven case study \cite{vanni2024visit}. This gravity-like specification aligns with standard gravity models \cite{schneider1959gravity,sen2012gravity,di2022gravity} and with recent AI-driven approaches \cite{guimera2020bayesian,cabanas2025human} that identify gravity-type OD laws as interpretable, parsimonious fits to mobility data.  We then generate fine- and coarse-grained networks via Monte Carlo simulation (using purpose-built code) and compare their estimates and plots with continuum-grained asymptotic analytical predictions; the excellent agreement across scales, therefore, corroborates our findings on occupancy and load-allocation (vacancy, coverage, overflow).
The paper closes with a discussion of potential applications and outlines how the framework could be extended in forthcoming research initiatives.

\section{Configurational ensembles of random pseudo-graphs}\label{sec_ensemblegraph}
	
In this section, we first introduce the general network framework for origin-destination flows, formulated as an urn model. 
Then,  we present two discrete configurational ensemble representations; namely, the fine-grained and coarse-grained formulations-as well as a continuum-based representation that incorporates a latent variable structure and draws upon the theory of inhomogeneous random graphs.

\subsection{Origin-Destination Urn Model}

A (directed pseudo-)graph \( \mathcal{G} = (V, E) \) consists of a set \( V \) of nodes (locations) and a collection \( E \) of directed edges (trips) between them. The graph allows for both self-loops and multiple edges between the same pair of nodes.
In collective mobility networks, such as origin-destination (OD) networks, nodes represent tiles of a spatial tessellation, while edges represent flows of people between these tiles.

Formally, we define a mobility network as a \textit{directed multigraph} \( G = (V, E) \), where:

\begin{itemize}
	\item[$\circ$] \( V \) is the set of nodes, corresponding to the tiles of a spatial tessellation;
	\item[$\circ$] \( V \times V \to \mathbb{N} \) is a function assigning to each pair of nodes the number of people moving between them (i.e., the mobility flow);
	\item[$\circ$] \( E = \{(i,j) \in V \times V \mid e(i,j) \neq 0\} \) is the set of directed edges, where multiple edges between nodes represent repeated flows between the same origin and destination.
\end{itemize}

This structure qualifies as a pseudograph since it also allows for \emph{self-loops}, representing intra-tile movements.
Each movement from node \( i \) to node \( j \) is modeled as a distinct edge, so that multiple edges encode repeated movements along the same path. This allows the network to precisely encode the number of trips or interactions between any two locations.

We model the generation of links using an urn scheme, where links are interpreted as balls launched from origin baskets to destination baskets (i.e., directed trips). The balls represent travelers and are grouped into \( L \) distinct categories, each corresponding to a different origin type or location.
These balls (or links) are randomly allocated into \( N \) destination boxes (nodes) as sketched in Fig.\ref{fig_balls2bins}. Since in this context each origin can also act as a destination, we assume \( L = N \). The total number of available balls is a multiple of the number of destinations, i.e., \( M = \mathfrak{n} N \), with \( \mathfrak{n} \geq 1 \), under a proportionate stratified random sampling framework\footnote{Proportionate stratified random sampling partitions a finite population into disjoint subgroups (strata) indexed by \( h = 1, \dots, L \), according to a categorical feature. Each stratum contains \( M_h \) units, so the total population size is \( M = \sum_{h=1}^L M_h \). A sample of size \( n \) is drawn such that each stratum contributes \( n_h = w_h n \) units, with \( w_h = M_h / M \) being the relative weight of stratum \( h \). Within each stratum, the sample is drawn randomly to ensure intra-stratum randomness. This guarantees that the sample reflects the population structure. In probabilistic terms, balls are drawn from baskets labeled by \( h \), each containing \( M_h = w_h N \) balls, where \( \sum_{h=1}^L w_h = \mathfrak{n} \) is the total relative mass. In the asymptotic regime \( \mathfrak{n} \to \infty \), each stratum contains an unbounded number of balls, and the sample size allocated to each stratum scales as \( n_h = w_h n \), preserving the stratified proportions in the limit.}.

\begin{figure}[!ht]
	\centering
	\begin{subfigure}[c]{0.9\textwidth}
		\centering
		\includegraphics[trim={0 0 0 5cm},clip,width=0.7\linewidth]{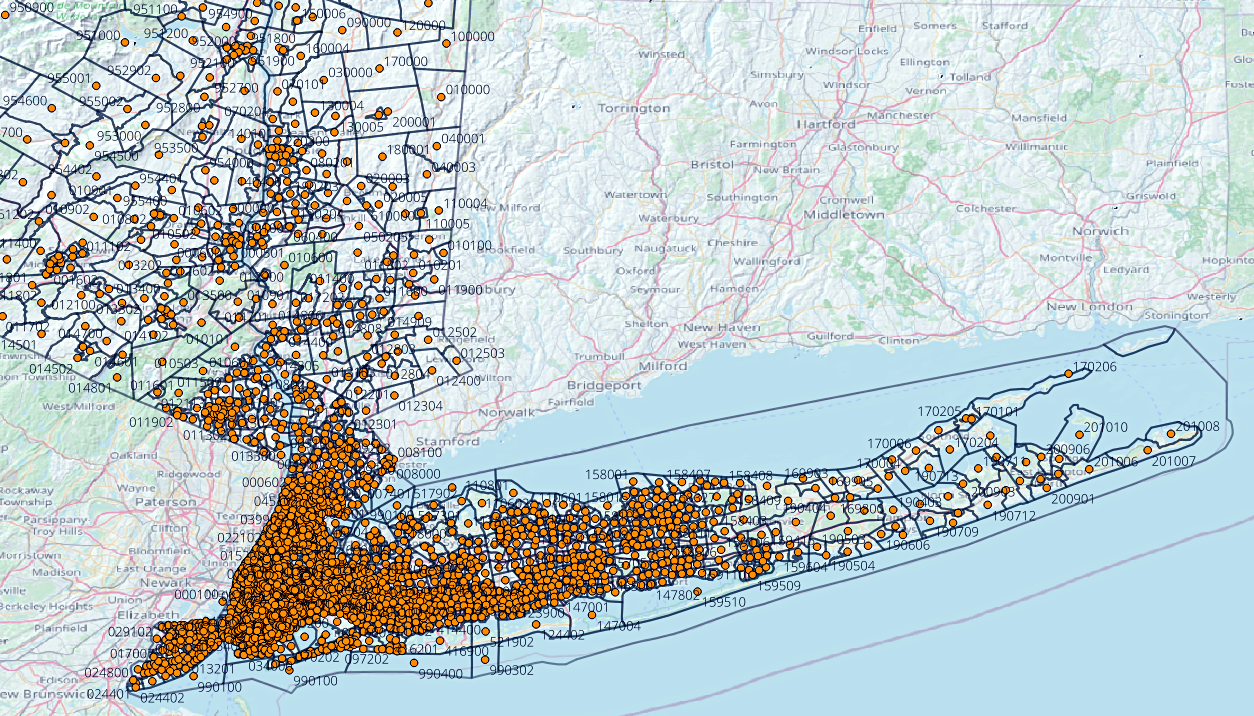}
		\caption{Centroids of census block groups as locations of a origin-destination network in New York state.}
	\end{subfigure} \\
	\vspace{1cm}
	\centering
	\begin{subfigure}[c]{0.9\textwidth}
		\centering
		\includegraphics[width=0.7\linewidth]{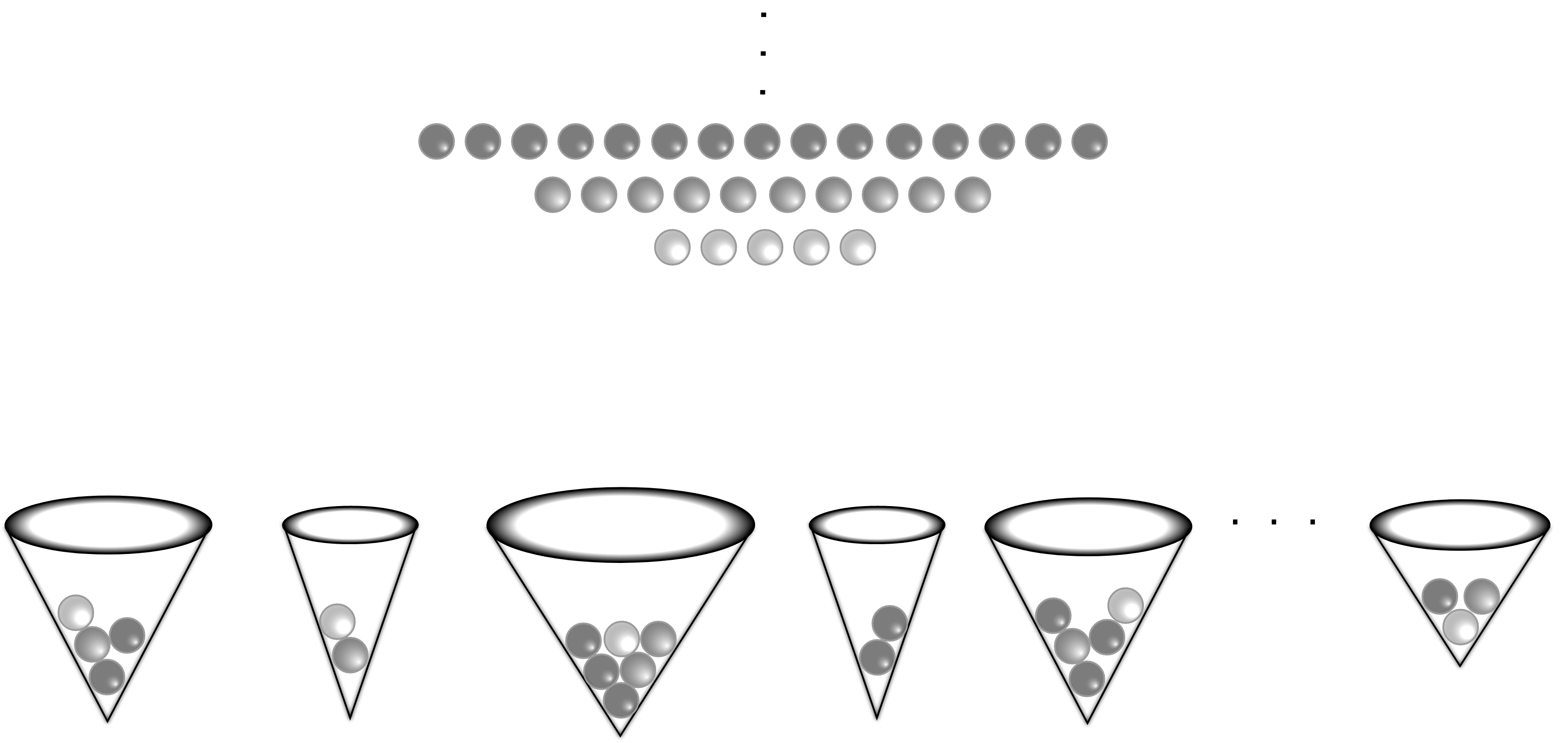}
		\caption{}
	\end{subfigure}
	\caption{Balls into bins representation. Balls represent trip-travelers and they can be of different types (shades of gray color).  The cone-bins represent locations as destinations and they have different probabilities to attract balls related to their entrance surfaces, and their capacity to attract visits is related to the volume of the bins.} \label{fig_balls2bins}
\end{figure}

Finally we define a random graph as a probability space \( \mathcal{G}(N,E) = (\Omega, \mathcal{F}, \mathbb{P}) \), where \( \Omega \) is the sample space of all graphs on \( N \) nodes, \( \mathbb{P} \) is the probability measure assigning to each graph \( G \in \Omega \) the probability \( \mathbb{P}(G) \), and \( \mathcal{F} \) is the associated $\sigma$-algebra, whose elements correspond to graph properties.
	In matrix form, the mobility network can be expressed through an {adjacency matrix}, where the entry \( a_{ij} \) counts the number of individual links (flows) from node \( i \) to node \( j \). Thus, the element \( a_{ij} \) represents the total number of trips or repeated movements from tile \( i \) to tile \( j \).

   \subsection{Coarse-grained and Fine-grained configurational probability}	
	

	We begin by defining a random graph ensemble in which the number of nodes \( N \) is fixed, and each pair of nodes \( (i,j) \) is connected independently with probability \( p_{ij} \). The total number of edges in the graph is not fixed a priori, but is instead a random variable. We refer to this construction as a \emph{coarse-grained probabilistic ensemble}, as it adopts a macroscopic perspective on graph formation referred to the {canonical ensemble} in the statistical physics of complex networks, as widely discussed in \cite{cimini2019statistical,anand2009entropy,park2004statistical}. Let us define an occupation probability	$p_{ij}$, so that the random graph $\mathcal{G}$ is a probability space over the set of graphs on the vertex set  determined by 
	$$	\text{Pr}[\{i, j\} \in \mathcal{G}] = p_{ij}$$ 	with these events mutually independent.
	Let us indicate with $\mathcal{A}$ the configuration of the mobility graph $\mathcal{G}$ when total number of observed trips is $n$, such configuration can be described in terms of the adjacency matrix of the graph which represent the OD table and where $a_{i,j}$ is the entry of the adjacency matrix indicating how many visits the destination $i$ received whose trips started from the origin $i$.

	\begin{proposition}[Coarse-grained Probabilistic Ensemble]
		A canonical random graph $\mathcal{G}_{N,\boldsymbol{p}}$ is a probability space on the set of graphs with $N$ nodes where  edges between pairs of nodes in the graph is added independently according to a probability vector $\boldsymbol{p}=\{p_{ij}\}, \forall i,j=1\ldots N$. 
		The probability to realize a member of configuration ensemble  graph  can be written as:
		\begin{equation}\label{eq_ensamble}
		\mathbb{P}_{\mathcal{C}}(\mathcal{A})=\prod_{i,j =1}^N \binom{\ell}{a_{ij}}p_{ij}^{a_{ij}}(1-p_{ij})^{\ell-a_{ij}} 
		\end{equation}
		where $\ell$, in the multi edges case, represents the number of attempts to access to the adjacency matrix entry $a_{ij}$, considering the set of multi graphs with at most $\ell$ links. 
	\end{proposition}
	\begin{proof}
		We have generalized the configuration ensemble graph for binary network in \cite{coolen2017generating,lee2018recent,squartini2017maximum,dorogovtsev2022nature} in order to consider multiple links in the adjacency matrix. In addition the variable $\ell$ takes into account possible multiple attempts to access the entry $a_{ij}$ in the adjacency matrix. The marginal probability of having a link can be written as $p_{ij}=\sum_{\mathcal{A}} a_{ij} 	\mathbb{P}_{\mathcal{C}}(\mathcal{A})$.
	\end{proof}
	
	
	We now adopt a fine-grained description of the same urn allocation problem by specifying the ensemble of random graphs in terms of the exact number of nodes \( N \) and the precise multiedge capacities \( m_{ij} \) allowed between each pair \( (i,j) \).  We consider the allocation of a total of \( n \) indistinguishable links (or ``balls'') to \( N\times N \) distinguishable node pairs \( (i,j) \), where each pair acts as a bin with an individual capacity constraint. In this combinatorial framework, a total of \( n \) links are drawn, and the configuration is such that each pair \( (i,j) \) has been selected exactly \( a_{ij} \leq m_{ij} \) times, where \( a_{ij} \) denotes the entry of the adjacency matrix.
	
	We refer to this construction as an \emph{enumerative random graph ensemble}, denoted by \( \mathcal{G}_{N,\boldsymbol{m}} \), which defines a uniform probability space over the set of all admissible unlabeled multigraphs consistent with the target edges  \( \{ m_{ij} \} \) and the total actual number of links \( n = \sum_{i,j} a_{ij} \). This formulation reminds to the {microcanonical ensemble} in statistical physics \cite{bianconi2007entropy,anand2009entropy,squartini2017maximum}, where the exact microstate is specified through a combinatorial allocation under constraints\footnote{Similarly with the micronanical ensemble, the fine-grained ensemble assigns equal probability to all microstates consistent with a set of constraints, but it is a more detailed version of the microcanonical ensemble since it can be viewed as a constrained multigraph allocation problem where the matrix $	m={m_{ij}}$ defines the pairwise link capacity constraints for the ensemble, specifying the maximum allowable number of links between each pair $(i,j)$ which  encodes the "skeleton" or "infrastructure" of the multigraph. }.

	%
	\begin{proposition}[Fine-grained Enumerative Ensemble]\label{prop_fine}
		Let $\mathcal{G}_{N,\boldsymbol{m}}$ be  the family of all  labeled  directed multigraphs with $N$ nodes and  exactly a number of links according to the vector $\boldsymbol{m}=\{m_{11},m_{12},\ldots m_{ij},\ldots,m_{NN}\}$ between the pairs $(i,j)$ $\forall i,j=1\ldots N$ which have  been selected $m_{ij}$ times. Then the probability to realize a member of configuration ensemble  graph  can be written as:
		\begin{equation}\label{eq_ensamblem}
		\mathbb{P}_\mathtt{m}(\mathcal{A}) = \frac{\binom{m_{11}}{a_{11}} \binom{m_{12}}{a_{12}} \cdots \binom{m_{ij}}{a_{ij}} \cdots \binom{m_{NN}}{a_{NN}}}{\binom{NM}{n}} =\frac{\prod_{i,j=1}^N \binom{\pi_{ij}NM}{a_{ij}}}{\binom{NM}{n}} 
		\end{equation}
		where  $\sum_{i,j=1}^N m_{ij}=NM$ and $\sum_{i,j=1}^N a_{ij} = n$  so that $\pi_{ij}=\frac{m_{ij}}{NM}=\text{const.}$  approximates the probability  to increase of one unit  the adjacency matrix entry $a_{ij}$.
		This expression gives the exact combinatorial probability of selecting a configuration \( \mathcal{A} \) by uniformly drawing \( n \) links from the total pool of \( NM \) available slots, under the capacity constraints \( a_{ij} \leq m_{ij} \).
		
		
	\end{proposition}
	\begin{proof}
		The goal is to distribute a total of \( n \) links among these available slots, where each slot \( (i,j) \) has capacity \( m_{ij} \), i.e., it can receive at most \( m_{ij} \) links. Each pair \( (i,j) \) is allowed up to \( m_{ij} \) parallel links. Let \( A = \{a_{ij}\} \) be the realization of such a configuration, where \( a_{ij} \leq m_{ij} \), and 
		\(
		\sum_{i,j} a_{ij} = n.
		\)
		The set $\mathcal{G}_{N,\boldsymbol{m}}$ consists of all labeled graphs with \( N \) nodes and an exact set of links \( \{m_{ij}\} \) between the node pairs \( (i,j) \), where each \( m_{ij} \) represents the number of links between node \( i \) and node \( j \). 
		The graph is constructed by distributing exactly \( n \) links (where \( n = \sum_{i,j} a_{ij} \)) among the possible \( \binom{2\binom{NM}{2}}{n} \) configurations (total number of ways to place \( n \) links among all possible pairs of nodes).
		The probability of realizing a specific graph configuration \( \mathcal{A} \) with the given link distribution \( \{m_{ij}\} \) is $P_m(\mathcal{A})$ as in the Proposition,
		The numerator \( \prod_{i,j} \binom{m_{ij}}{a_{ij}} \) is the number of ways to choose \( a_{ij} \) links among the \( m_{ij} \) available between each pair \( (i,j) \).
		The denominator \( \binom{NM}{n} \) is the total number of ways to assign \( n \) links among all \( NM \) stubs, assuming all \( m_{ij} \) are distributed in a single bin of size \( NM \), where
		\(
		NM = \sum_{i,j} m_{ij}
		\).  A similar framework derivation has been obtained in \cite{casiraghi2021configuration}.
	\end{proof}
	
	Among the twos, the coarse-grained random graph ensemble \( \mathcal{G}_{N,\boldsymbol{p}} \), characterized by independent edge probabilities \( \{p_{ij}\} \), is more amenable to analytical treatment due to the independence structure of edges. In contrast, the  fine-grained ensemble \( \mathcal{G}_{N,\boldsymbol{m}} \), which imposes fixed multiedge capacities \( \{m_{ij}\} \), lacks such independence and exhibits combinatorial dependencies among edges. 
	
	The two configurational ensembles provide complementary perspectives that are both epistemic and ontological according to which mechanism is believed to be more phenomenologically realistic.
	The fine-grained ensemble approach assumes complete information about the allowable structure, akin to knowing the precise result of an edge-by-edge rewiring of allocation mechanism: the randomness in this ensemble is purely enumerative i.e., over the finite set of admissible configurations. On the other hand, the coarse-grained ensemble approach represents a probabilistic inference framework under partial information, where the connection probabilities  represent  edge structure preferences: the model reflects uncertainty or incomplete knowledge about the system and the generative process is stochastic and reflects an epistemic limitation: you don't observe the full allocation mechanism, only its probabilistic tendencies.

	Nonetheless, the two ensembles are closely related, and under suitable asymptotic regimes, they yield equivalent statistical behavior:
	
	\begin{theorem}[Asymptotic Equivalence of Ensembles]\label{th_macromicro}
		Let \( \mathcal{A} = \{a_{ij}\} \) be an admissible graph configuration with total number of links \( n = \sum_{i,j} a_{ij} \). Consider the following asymptotic regime of sparse and large graphs:
		\begin{itemize}
			\item[$\circ$] coarse-grained ensemble: \( \ell \gg 1 \), \( p_{ij} \ll 1 \), and \( \ell \sum_{i,j} p_{ij} \ll 1 \),
			\item[$\circ$] fine-grained ensemble: \( n \gg 1 \), \( N \to \infty \), and \( n \ll \mathfrak{n}  \).
		\end{itemize}
		Then, under these sparse conditions, the conditional probability of \( \mathcal{A} \) in the coarse-grained ensemble, given the total number of links, is asymptotically equivalent to the fine-grained probability:
		\begin{equation}
		\mathbb{P}_{\mathcal{C}}\Big(\mathcal{A}\,\Big|\sum_{i,j=1}^N a_{ij} = n\Big) \sim \mathbb{P}_\mathtt{m}(\mathcal{A}) = n! \prod_{i,j=1}^N \frac{\pi_{ij}^{a_{ij}}}{a_{ij}!}, \qquad \text{ where } \; 	\pi_{ij} := \frac{m_{ij}}{NM} = \frac{p_{ij}}{\sum_{i,j} p_{ij}}.
		\end{equation}
	\end{theorem}
	
	\begin{proof}
In the coarse-grained case, conditioning on $\sum_{i,j} a_{ij}=n$, we have 
$\mathbb{P}_{\mathcal{C}}(\mathcal{A}\mid n)=\mathbb{P}_{\mathcal{C}}(\mathcal{A})/\mathbb{P}_{\mathcal{C}}(n)$. 
For $\ell\to\infty$ with $p_{ij}\ll 1$ and $\ell p_{ij}<\infty$, Stirling and Taylor expansions give:
\[
\mathbb{P}_{\mathcal{C}}(\mathcal{A}) \sim \prod_{i,j}\frac{(\ell p_{ij})^{a_{ij}}}{a_{ij}!}\,e^{-\ell p_{ij}}, 
\qquad 
\mathbb{P}_{\mathcal{C}}(n)\sim \frac{1}{n!}e^{-\ell\sum p_{ij}}(\ell\sum p_{ij})^n.
\]
Hence
\[
\mathbb{P}_{\mathcal{C}}(\mathcal{A}\mid n)=n!\prod_{i,j}\frac{(p_{ij}/\sum p_{ij})^{a_{ij}}}{a_{ij}!}\left(1+\mathcal{O}\!\left(\tfrac{1}{\ell}\right)\right),
\]
which tends to a multinomial law with weights $p_{ij}$.  

In the fine-grained case, 
\[
\mathbb{P}_m(\mathcal{A})=\frac{\prod_{i,j}\binom{m_{ij}}{a_{ij}}}{\binom{NM}{n}}. 
\]
Using Stirling’s formula and expanding the ratio of factorials yields
\[
\mathbb{P}_m(\mathcal{A})\approx n!\prod_{i,j}\frac{1}{a_{ij}!}\left(\frac{m_{ij}}{NM}\right)^{a_{ij}}
\exp\!\left(-\tfrac{1}{2}\sum_{i,j}\frac{a_{ij}(a_{ij}-1)}{m_{ij}}+\tfrac{n^2}{2NM}\right).
\]
Under sparsity ($n\ll NM$, $a_{ij}\ll m_{ij}$), corrections vanish and we obtain
\[
\mathbb{P}_m(\mathcal{A})\sim n!\prod_{i,j}\frac{\pi_{ij}^{a_{ij}}}{a_{ij}!}, \qquad \pi_{ij}=\frac{m_{ij}}{NM},
\]
which again is multinomial.  

For the complete derivations see Appendix~\ref{app_th_macromicro}.	
\end{proof}
	
From the previous Theorem we can derive the following two important consequences.	
	
	\begin{corollary}\label{coro_ell}
		The number of attempts to access the adjacency matrix entries is proportional to the expected total number of links in the unconditional case  as 
		\begin{equation}
		\ell = \langle n\rangle /\sum_{i,j=1}^N p_{ij}.
		\end{equation}
	\end{corollary}
	\begin{proof}
		This comes from the observation that  the expected value of each entry \( a_{ij} \) in the adjacency matrix \( \mathcal{A} \) is $\mathbb{E}[a_{ij}] = \ell \cdot p_{ij}$.
		Thus, the unconditional expected adjacency matrix \( \mathbb{E}[\mathcal{A}] \) has entries \( \ell \cdot p_{ij} \) for each pair \( (i, j) \).
		Under the condition that the total number of links is fixed at \( n \), i.e., \( \sum_{i,j} a_{ij} = n \), the conditional expected value of each \( a_{ij} \) given \( \sum_{i,j} a_{ij} = n \) is:
		\[
		\mathbb{E}[a_{ij} \mid \sum_{i,j} a_{ij} = n] = n \cdot \frac{p_{ij}}{\sum_{i,j} p_{ij}}=\mathbb{E}[a_{ij}] \cdot \frac{n}{\langle{n}\rangle },
		\]
		where \( \langle{n}\rangle = \mathbb{E}[n]= \ell \sum_{i,j} p_{ij} \) is the expected total number of links in the unconditional case. This shows that the conditional expectation rescales the unconditional expectation by the factor \( \frac{n}{\langle{n}\rangle} \) to satisfy the fixed total constraint.
	\end{proof}
	In particular, for homogeneous random graphs, $\langle n\rangle =\ell \sum_{i,j=1}^N p_{ij}= p N^2$ since $\ell=1$ and $p_{ij}=p$.
	
	\begin{corollary}[Homogeneous random graph]
		Let us consider the case of homogeneous binary random graph where $p_{ij}=p$ $\forall i,j=1\ldots N$. Then, the probability in eq.\eqref{eq_ensamble} to generate a canonical random graph $\mathcal{G}_{n,\boldsymbol{p}}$ in a configuration $\mathcal{A}$  is  equivalent to the Gilbert derivation of binomial random graphs since: 
		\begin{align*}
		\mathbb{P}_{\mathcal{C}}(\mathcal{A})=p^n(1-p)^{N^2-n}
		\end{align*} 
		where $n=\sum_{ij}a_{ij}$. 
		
		Furthermore, the  Erd\"os and R\'enyi approach to generate a micro-canonical random graph $\mathcal{G}_{n,\boldsymbol{m}}$ is obtained when the probability eq.\eqref{eq_ensamblem} for directed and binary graphs (with self-loops) becomes:
		$$\mathbb{P}_\mathtt{m}(\mathcal{A})=\frac{1}{\binom{N^2}{n}}=\mathbb{P}_{\mathcal{C}}(\mathcal{A}|n),$$ which denotes such a
		random graph  as the uniform (or homogeneous) random graph. Finally $\mathbb{E}[n]= pN^2 $.
	\end{corollary}
	\begin{proof}
		From eq.\eqref{eq_ensamble} In the coarse-grained ensemble of random graphs, we consider only one entry is possible since the graph is a binary simple graph (no multi edges), so that  $\ell=1$. Then  since $M=N$, i.e. $h=1$: $$\mathbb{P}_{\mathcal{C}}(\mathcal{A})=\prod_{i,j} p^{a_{ij}}(1-p)^{1-a_{ij}} = p^{\sum_{i,j} a_{ij}}(1-p)^{N(N-1)-\sum_{i,j} a_{ij}}=p^n(1-p)^{N^2-n}$$ which is the Gilbert graph derivation as in  \cite{janson2011random,bollobas1998random}. Consequently, the conditional coarse-grained distribution is $$\mathbb{P}_{\mathcal{C}}(\mathcal{A}|n)=\frac{\mathbb{P}_{\mathcal{C}}(\mathcal{A})}{\binom{N^2}{n}p^n(1-p)^{N^2-n}}$$
		
		On the other side, in the fine-grained ensemble of random graphs, we have $m_{ij}=\{0,1\}$, so that:
		$$\mathbb{P}_\mathtt{m}(\mathcal{A})=\binom{N^2}{n}^{-1}\equiv \mathbb{P}_{\mathcal{C}}(\mathcal{A}|n)$$
	\end{proof}



\subsection{Continuum-Grained Configurational Probability}

The continuum-grained ensemble generalizes the coarse-grained configuration model by replacing the discrete set of edge probabilities \( \{p_{ij}\} \) with a continuous kernel function \( \kappa \), which serves as a limit object for large sparse graphs. In this limit, \( \kappa \) defines a graphon, interpreted as a continuous adjacency function over a latent space, encoding the connection intensity or probability between nodes characterized by latent features.
To formalize this, we adopt a latent-variable framework for the urn-based mobility graph, following the approach in \cite{vanni2024visit}. Each node (or bin) \( i \in \{1,\ldots,N\} \) is endowed with two independent latent variables:
$
x_i \in \mathcal{X},  y_i \in \mathcal{Y},
$
where \( x_i \) denotes the destination-type (capturing the attractiveness of node \( i \) as a destination) and \( y_i \) the origin-type (representing the population or demand type generating trips from node \( i \)). These latent variables are drawn independently from probability measures \( \mu_{\mathcal{X}} \) and \( \mu_{\mathcal{Y}} \), respectively, defined over Polish spaces \( (\mathcal{X}, \mathcal{B}_{\mathcal{X}}) \) and \( (\mathcal{Y}, \mathcal{B}_{\mathcal{Y}}) \), with Borel sigma-algebras and densities absolutely continuous with respect to Lebesgue measure:
\[
x_1, \dots, x_N \overset{\text{iid}}{\sim} \mu_{\mathcal{X}}, \qquad y_1, \dots, y_N \overset{\text{iid}}{\sim} \mu_{\mathcal{Y}}, \qquad \text{with } x_i \perp y_j.
\]

In the spirit of inhomogeneous random graph theory \cite{bollobas2007phase,turova2011survey},
 define a directed random graph model \( G^{\mathcal{V}}(N,\kappa) \) over the vertex space \( \mathcal{V} = (\mathcal{X}, \mathcal{Y}, \mu_{\mathcal{X}}, \mu_{\mathcal{Y}}) \), 
 For each \( N \), the collection of vertices in the graph \( G^{\mathcal{V}}(N, \kappa) \) is formed by either a deterministic or random sequence \( x_1, \ldots, x_N \) of points within \( S \), such that for any \( \mu \)-continuity set \( A \subseteq \mathcal{X} \) and \( B \subseteq \mathcal{Y} \),
 \[
 \frac{\#\{i : x_i \in A\}}{N} \overset{P}{\to} \mu_{\mathcal{X}}(A) \quad , \quad	\frac{\#\{i : y_i \in B\}}{N} \overset{P}{\to} \mu_\mathcal{Y} (B)
 \]
 Moreover, for each \( i,j \in \{1, \ldots, N\} \), an edge from \( j \) to \( i \) is formed independently with probability:
\[
p_{ij} = \min\left\{ \frac{\kappa(x_i, y_j)}{N},\, 1 \right\},
\]
for some symmetric or directed kernel \( \kappa : \mathcal{X} \times \mathcal{Y} \to [0, \infty) \), assumed to be measurable and integrable, i.e., \( \kappa \in L^1(\mu_{\mathcal{X}} \otimes \mu_{\mathcal{Y}}) \).

The average number of links per node (edge density) in \( G^{\mathcal{V}}(N,\kappa) \) satisfies:
\begin{equation}\label{eq_edgedense}
\frac{1}{N} \mathbb{E}[e(G^{\mathcal{V}}(N,\kappa))] = \frac{\ell}{N^2}\sum_{i,j=1}^{N} \big( {\kappa(x_i, y_j)}\wedge 1 \big)\to \ell \iint_{\mathcal{X} \times \mathcal{Y}} \kappa(x,y) \, d\mu_{\mathcal{X}}(x) \, d\mu_{\mathcal{Y}}(y),
\end{equation}
as \( N \to \infty \), interpreting the right-hand side as the population-average link intensity — the continuum analogue of average degree in sparse graphs.

To capture the empirical structure of such graphs, define the empirical graphon in sparse regime \cite{glasscock2015graphon}:
\begin{equation}
\mathcal{\tilde{A}}_N(x,y) := \frac{1}{N} \sum_{i,j=1}^N a^{(\ell)}_{ij} \, \delta_{x_i}(x) \, \delta_{y_j}(y),
\end{equation}
where \( a^{(\ell)}_{ij} \in [0, \infty) \) is the multi edge from \( j \) to \( i \) and in order to examine the sparse limit can be defined as:
\[
a_{ij}^{(\ell)} := \operatorname{Binomial} \left( \ell,\, \tfrac{\kappa(x_i,y_j)}{\ell N} \right).
\] 
As \( N \to \infty \),  this empirical measure converges weakly (in distribution) to a deterministic limit \( \mathcal{\tilde{A}}(x,y) \in L^1(\mathcal{X} \times \mathcal{Y}) \), which plays the role of a continuous adjacency matrix over the latent type space \( \mathcal{X} \times \mathcal{Y} \).

In the continuum-weighted regime, the ensemble distribution of such graphs is characterized by a probability density \( P(\mathcal{\tilde{A}}) \) over admissible limit graphons \( \mathcal{\tilde{A}} \), encoding the statistical variability of link weights across the latent space.

\begin{proposition}[Continuum-Grained  Ensemble]\label{conti_ensamble}
	Let \( \tilde{\mathcal{G}}_{N,\kappa} \) be a probability space over multi-edge random graphs with \( N \) nodes, where each node \( i \) is associated with independent latent types \( (x_i, y_i) \) drawn from \( \mu_{\mathcal{X}} \otimes \mu_{\mathcal{Y}} \). For each ordered pair \( (i,j) \), an edge from \( j \) to \( i \) is formed independently according to the intensity kernel \( \kappa \in L^1(\mu_{\mathcal{X}} \otimes \mu_{\mathcal{Y}}) \). 
	
	Then, the empirical graphon
	\(
	\tilde{\mathcal{A}}_N(x,y)
	\)
	converges in distribution to a limit \( \tilde{\mathcal{A}} \in L^1(\mathcal{X} \times \mathcal{Y}) \), and satisfies a large deviations principle in the weak topology of \( L^1 \):
	\[
	\mathbb{P}_{\aleph}(\tilde{\mathcal{A}}) = \Pr_N\{\tilde{\mathcal{A}}_N \approx \tilde{\mathcal{A}} \} \asymp \exp\left[ -N \, \mathcal{I}(\tilde{\mathcal{A}}) + o(N) \right],
	\]
	where the rate functional \( \mathcal{I} : L^1 \to [0,\infty] \) is given by:
	\[
	\mathcal{I}(\tilde{\mathcal{A}}) = D(\tilde{\mathcal{A}} \Vert \ell \kappa) +\ell  \|\kappa\|_1 - \|\tilde{\mathcal{A}}\|_1,
	\]
where \( D(\tilde{\mathcal{A}} \Vert \ell \kappa)\)  is the relative entropy, \( \|\tilde{\mathcal{A}}\|_1 \) is the total observed edge mass, and \( \ell \|\kappa\|_1 \) is the expected total edge mass under the model.

\end{proposition}
\begin{proof}
	Since \( \ell \to \infty \), we have \( a_{ij}^{(\ell)} \xrightarrow{\mathsf{d}} \text{Poisson}(\kappa(x_i,y_j)/N) \) by the classical Poisson limit of Binomials.
	The Sanov theorem for product Poisson measures (or alternatively, a G\"{a}rtner-Ellis theorem under suitable regularity) \cite{dembo2009large,leonard2000large} implies that \( \tilde{\mathcal{A}}_N \) satisfies a large deviations principle in the weak topology of \( L^1(\mathcal{X} \times \mathcal{Y}) \), at speed \( N \), with good rate functional:
	\[
	\mathcal{I}(\tilde{\mathcal{A}}) = \iint_{\mathcal{X} \times \mathcal{Y}} \left[ \tilde{\mathcal{A}}(x,y) \log\left( \frac{\tilde{\mathcal{A}}(x,y)}{\ell \kappa(x,y)} \right) - \tilde{\mathcal{A}}(x,y) + \ell \kappa(x,y) \right] \, d\mu_{\mathcal{X}}(x) \, d\mu_{\mathcal{Y}}(y).
	\]
	This is the relative entropy between \( \tilde{\mathcal{A}} \) and \( \kappa \), and it is nonnegative, lower semicontinuous, and vanishes if and only if \( \tilde{\mathcal{A}} =\ell  \kappa \) almost everywhere.	
	Therefore, for any measurable set \( A \subseteq L^1 \), the large deviations principle holds:
	\[
	\mathbb{P}_{\aleph}(\tilde{\mathcal{A}}_N \in A) \asymp \exp\left( -N \cdot \inf_{\tilde{\mathcal{A}} \in A} \mathcal{I}(\tilde{\mathcal{A}}) \right).
	\]
	
\end{proof}

In probabilistic ensembles, \( \|\mathcal{\mathcal{\tilde{A}}}\|_1 \)  represents the {realization} of edge mass, while \( \ell \|\kappa\|_1 \) gives the {expected} edge mass under the model. The rate functional \( \mathcal{I}(\mathcal{\tilde{A}}) \) quantifies how different the observed link intensity \( \mathcal{\tilde{A}} \) is from the model kernel \( \kappa \), both in terms of information content and probabilistic likelihood.

Because the empirical graphon $\widetilde{\mathcal{A}}_N(x, y)$ in sparse regime is built with Dirac deltas,
evaluating it on the \( (i,j) \)-th latent pair simply divides the discrete multiplicity by \( N \)
$
\widetilde{\mathcal{A}}_N(x_i, y_j) = {a_{ij}}/{N}$.
That is, each entry \(a_{ij}\) records an exact count of links, while \(\tilde{\mathcal{A}}_{N}(x_i,y_j)=a_{ij}/N\) is a link density as the same count rescaled by the network size \(N\).
The first description keeps the data as raw integers appropriate for a finite set of nodes; the second turns it into a flow density that remains \(\mathcal{O}(1)\) as \(N\) grows.
Hence the two quantities carry exactly the same information, one at a discrete scale, the other at a continuum-normalized scale.

Fixing the total number of links, we get the graph constraint to be:
\[
\sum_{i,j} a_{ij} = n \quad \Longleftrightarrow \quad \|\widetilde{\mathcal{A}}_N\|_1 
= \iint \widetilde{\mathcal{A}}_N(x, y)\, d\mu_{\mathcal{X}}(x)\, d\mu_{\mathcal{Y}}(y) = \frac{n}{N},
\]
so the {total} number of links \( n = \mathcal{O}(N) \) in the sparse continuum model becomes an \( \mathcal{O}(1) \) “mass” for the rescaled graphon.

Under suitable regularity and normalization,  in the large-\( N \) limit, the coarse-grained probability \( \mathbb{P}_{\mathcal{C}}(\mathcal{A}) \) of observing a configuration converging to \( \tilde{\mathcal{A}} \) behaves according to the following statement:

\begin{theorem}[Large Deviation Graphon under Fixed Mass]\label{theo_graphon}
Let \( \tilde{\mathcal{A}} : \mathcal{X} \times \mathcal{Y} \to [0, \infty) \) be an admissible empirical graphon, and let \( \kappa : \mathcal{X} \times \mathcal{Y} \to [0, \infty) \) be the model kernel. Suppose both functions are integrable with respect to the product measure \( \mu \otimes \nu \).  Under the constraint the total mass of the graph is fixed as:
\[
\|\widetilde{\mathcal{A}}_N\|_1 
=\ell \|\widetilde{\kappa}\|_1 
=\frac{n}{N}
\]
then, the coarse-grained model converges to the continuum configurational ensemble:
\[
\mathbb{P}_{\mathcal{C}}(\mathcal{A} \big | n)
\quad \longleftrightarrow \quad
\mathbb{P}_{\aleph}(\tilde{\mathcal{A}}  \big | \|\tilde{\mathcal{A}}\|_1 = n/N)
\asymp
\exp\left[ -N \cdot D(\tilde{\mathcal{A}} \| \ell \kappa) \right].
\]

where we can define the probability density as:
\begin{equation}
\pi_{ij}=\frac{\ell N }{n}{\kappa(x_i,y_j)} 
\end{equation}

 with the kernel evaluated at the sampled points $(x_i,y_j)$.
\end{theorem}
\begin{proof}
	This is just a special case of Proposition \ref{conti_ensamble}, and the probability density can be recovered as:
	\begin{equation*}
	\pi_{ij}=\frac{p_{ij}}{\sum_{ij}p_{ij}}\to \frac{\ell \kappa(x,y)}{\ell \| \kappa \|_1} \Rightarrow  \pi_{ij}=\frac{\ell N }{n}{\kappa(x_i,y_j)} 
	\end{equation*} 	
	Let \(A=\{a_{ij}\}_{1\le i,j\le N}\) be a discrete configuration with total mass \(n=\sum_{i,j}a_{ij}\) and let its continuum rescaling be \(\widetilde A:=A/N\); by definition the \(L^{1}\)-mass of the rescaled matrix is \(\|\widetilde A\|_{1}=\tfrac{1}{N}\sum_{i,j}a_{ij}=n/N\). In the large-deviation expression \(\exp[-N\,D(\widetilde A\|\ell\kappa)]\) the reference profile \(\ell\kappa\) must carry the same mass as \(\widetilde A\); hence one sets \(\ell:=\|\widetilde A\|_{1}\). Substituting the preceding identity for \(\|\widetilde A\|_{1}\) yields \(\ell=n/N\), which rearranges to the equality \(n=\ell N\). Thus the relation follows directly from the definition of \(\ell\) as the empirical average load per node. Starting from the coarse-grained probability \(\mathbb{P}_C(A \mid n) = \frac{n!}{\prod_{i,j} a_{ij}!} \prod_{i,j} \pi_{ij}^{a_{ij}}\), and applying Stirling’s approximation \(\log a_{ij}! \approx a_{ij} \log a_{ij} - a_{ij}\), one obtains \(\log \mathbb{P}_C(A \mid n) \approx -n D(A/n \| \pi)\), where \(\pi_{ij} = \ell \kappa(x_i, y_j)/\sum_{i,j} \kappa(x_i, y_j)\). Substituting \(\ell = n/N\) gives \(\pi_{ij} = \kappa(x_i, y_j)/\sum_{i,j} \kappa(x_i, y_j)\), and changing variables from \(A/n\) to \(\widetilde A = A/N\) yields \(D(A/n \| \pi) = D(\widetilde A \| \ell\kappa)\), so that \(\mathbb{P}_C(A \mid n) \asymp \exp[-N D(\widetilde A \| \ell\kappa)] = \mathbb{P}_{\aleph}(\widetilde A \big | \|\widetilde A\|_1 = \ell)\), proving the equivalence between the coarse-grained and continuum-grained probability representations.

\end{proof}

		\section{Occupancy Problem \& Degree Distributions}\label{sec_occupancy}
		In occupancy theory, objects (balls) of various types are randomly assigned to containers (boxes or bins), which may differ in type and capacity. Depending on the setting, bins can have finite or infinite capacity, and the balls can be either distinguishable (ordered) or indistinguishable (unordered) (cf.\ \cite{johnson1977urn,feller1971introduction}). This framework allows for the study of several important random variables, such as \( N_k \), the number of urns that contain exactly \( k \) balls.  		
		Urn models provide a natural and general framework for formulating and analyzing occupancy problems, with many classical problems in this area arising as specific cases of urn-based formulations.
		
		The allocations of visiting patterns in human mobility graphs can be formalized in combinatorial terms as an occupation problem of balls into bins, where bins (nodes) represent the destination locations and balls (links) are the travelers originating from different locations (colors). So each mobility graph can be seen as a balls-into-bins setting. The graph is directed since  locations are the nodes of the networks seen as origin destinations (bins or boxes) and as origin locations as well (which define the ball type as colors).

\subsection{Fine-grained occupancy distribution}		
		We aim to derive the degree distribution of the origin-destination mobility graph \(\mathcal{G}\) by modeling it as an occupancy problem. Specifically, consider \(n\) balls selected from a total of \(M\) balls, which are independently thrown into \(N\) bins. 	If there are k balls in a particular box, then $k$ is the occupancy number of the box. Occupancy problems arise when analyzing the distribution of \(N_k\), which represents the number of bins that contain exactly \(k\) balls, where \(k = 0, 1, \ldots, n\). Consequently, \( N_k \) is the number of boxes containing exactly \( k \) balls, then  the occupancy distribution $P_n(k)=\frac{\mathbb{E}[N_k]}{N}$ is  the fraction of the expected number of boxes with occupancy number $k$.
		
		From a microscopic perspective (fine-grained  ensemble), each ball is placed  into one of the bins sequentially, with the proportion of balls being assigned to the \( i \)-th bin given by \( \pi_i=\sum_{j}m_{ij}/NM \), so that the proportions will sum to 1, i.e. \( \pi_1 + \pi_2 + \cdots + \pi_N = 1 \).

		\begin{theorem}[Fine-grained occupancy]\label{th_microcanoccu}
			Let \( Z^{(n)}_i \), the occupancy umber, which denotes the number of balls placed into box \( i \) when \( n \) total  balls (visits) have been drawn, regardless of type, from an urn containing \( M \) balls. The joint probability mass function (pmf) of the random variables \( (Z^{(n)}_1, Z^{(n)}_2, \dots, Z^{(n)}_N) \) is the joint distribution of the occupancy numbers which is given by:			
			\[
			\Pr\bigg[\bigcap_{i=1}^{N}(Z^{(n)}_{i}=k_{i})\bigg] = P_m(k_1, k_2, \ldots, k_N) = \tfrac{n!}{\prod_{i=1}^N k_i!} \prod_{i=1}^N \pi_i^{k_i} \cdot \tfrac{\prod_{j=0}^{k_i-1} \left(1 - \frac{j}{\pi_i NM} \right)}{\prod_{j=0}^{n-1} \left(1 - \frac{j}{NM} \right)}
			\]
			
			where \( k_1, k_2, \dots, k_N \) are nonnegative integers such that \( \sum_{i=1}^{N} k_i = n \), and \( \pi_i \) represents the proportion of the total number of balls placed into box \( i \), given by $
			\pi_i = \sum_{j=1}^{N} \pi_{ij}$.
			
			For large \( N \), the expected number of boxes containing exactly \( k \) balls, denoted by \( \mathbb{E}[N_k] \), is given by:
			
			\[
			\mathbb{E}[N_k] = \sum_{i=1}^N \Pr(Z^{(n)}_i=k) = \sum_{i=1}^N \frac{n!}{k! (n - k)!} \pi_i^k (1 - \pi_i)^{n - k} \left( 1+ \mathcal{O}\left(\tfrac{n^2}{\mathfrak{n}N^2}\right) \right)
			\]
			
			where \(\Pr(Z_i^{(n)}=k) \) is the marginal probability that box \( i \) contains exactly \( k \) balls.
		\end{theorem}
		
		\begin{proof}
			The proof is shown in the appendix~\ref{app_th_microcanoccu}
		\end{proof}
This expression represents the expected number of boxes with exactly \( k \) balls, incorporating the finite-population corrections from the hypergeometric-like structure,  providing a precise approximation of the expected count with an explicit big-O error bound. 

In the multinomial limit of $\mathfrak{n}\to \infty$, the framework can be treated as  a classic  "balls into bins" problem of finding the distribution of $N_k$ when there are $N$ urns and each ball will fall  into any 	particular bin $j$ according to constant proportions  $\pi_j$ and  independently from other bins.
So the following  corollary recovers the occupancy distribution in the condition of pure multinomial bin-into-bins setting.
	\begin{remark}[Asymptotic limit]\label{coro_asym_fin}
			Let \( \tilde{Z}^{(n)}_i \) denote the number of balls placed into box \( i \) when \( n \) total balls (visits) have been drawn, regardless of type, from an  infinite reservoir of undistinguishable balls i.e. $ \mathfrak{n} \to \infty$. 
	 Then, the occupancy distribution  is a pure binomial mixture:
		\begin{align}
		P_n(k)=\tfrac{\mathbb{E}[N_k]}{N}= & \tfrac{1}{N}\sum_{i=1}^{m} \binom{n}{k}\pi_i^k (1 - \pi_i)^{n-k} =\tfrac{1}{N} \sum_{i=1}^N\frac{\left(n\pi_i \right)^k}{k!}e^{-n\pi_i}\left( 1+\mathcal{O}(n\pi_i^2)\right)
		\end{align}
		which represents the	expected probability mass function (PMF) of the visit distribution, under the condition of large and sparse networks $N\to \infty$ and $n\ll \mathfrak{n}$.
	\end{remark}	
\begin{proof}
In this case, we allow the boxes to have heterogeneous occupancy proportions $\pi_i$. 
Following the generating function approach to occupancy distributions presented in \cite[Ch.~3]{johnson1977urn}, we derive the corresponding distribution. 
Since the proof is original and involves non-trivial steps that go beyond the standard results in the reference, we present the full derivation in Appendix~\ref{app_asym_fine}.

\end{proof}

Let us notice that in balls-and-bins problems, a primary challenge arises from the dependencies among the bins: if we know the number of balls in the first \( N - 1 \) bins, the number of balls in the final bin is fully determined. This dependency complicates analysis since the bin loads are not independent, and independent random variables are typically far easier to work with. A powerful technique to address these dependencies is \textit{Poissonization}, which introduces independent Poisson variables to approximate the original system.
The key distinction between throwing \( M \) balls with independent probabilities \( p_{ij} \) and assigning each bin an independently Poisson-distributed count with mean \( \sum_{j=1}^M p_{ij} \) is that, in the first case, we have exactly \( M \) balls in total, whereas in the Poissonized version, \( M \) only represents the expected total number of balls across bins. However, by conditioning on the event that the sum of the independent Poisson variables equals \( M \), we can approximate the system as if we were tossing \( M \) balls into \( N \) bins independently according to the probabilities \( p_{ij} \).

\subsection{Coarse-grained occupancy distribution}	
		From a macroscopic (coarse-grained) perspective, we consider throwing balls into \( N \) bins, where each ball \( j \) lands in bin \( i \) according to the probabilities \( p_{ij} \), for \( i,j = 1, \ldots, N \). And we launched a total of \( n \) balls.
		
\begin{theorem}[Coarse-grained occupancy]\label{th_canoccu}
			Let \( Y_i^{(\ell)} \) represent the number of balls in the \( i \)-th bin after $\ell$ attempts, where \( 1 \leq i \leq N \), then the joint probability mass function of the macroscopic random variables \( (Y_1^{(\ell)}, \ldots, Y_N^{(\ell)}) \), conditioned on \( \sum_{i} Y_i^{(\ell)} = n \), is the joint distribution of the occupancy numbers which is given by:
			\begin{align*}
			\Pr\bigg[\bigcap_{i=1}^{N}(Y^{(\ell)}_{i}=k_{i})\Big | \sum_{i} Y_i^{(\ell)} = n \bigg] 
			=& P_C(k_1, k_2, \dots, k_N \big | \sum_{i=1}^N k_i=n)  \\
			=&   \tfrac{n!}{\prod_{i=1}^N k_i!} \prod_{i,j=1}^N \bigg( \tfrac{p_{ij}}{\sum_{j=1}^N p_{ij}} \bigg)^{k_i} 
			\left[ 1 + \mathcal{O}\left(  \tfrac{k_i^2}{\ell} \right) \right],
			\end{align*}
			where \( k_1, k_2, \dots, k_N \) are nonnegative integers such that \( \sum_{i=1}^{N} k_i = n \).
			
			In particular, for large and sparse networks, the expected number of bins containing exactly \( k \) balls, denoted by \( \mathbb{E}[N_k] \), is given by:
			\begin{align}
			\mathbb{E}[N_k] &= \sum_{i=1}^N \Pr(Y^{(\ell)}_i=k \mid \sum Y^{(\ell)}_i=n) 
			= \sum_{i=1}^N \binom{n}{k} \nu_i^k (1-\nu_i)^{n-k} 
			\bigg[ 1 + \mathcal{O}\left( \tfrac{k^2}{\ell} \right) \bigg] 
			\end{align}
			where \( \nu_i = \frac{\sum_j p_{ij}}{\sum_i\sum_{j}p_{ij}} \) represents the (normalized) probability of a ball landing in bin \( i \).
\end{theorem}
		\begin{proof}
			The proof is presented in the appendix\ref{app_th_canoccu}.
		\end{proof}

In statistical mechanics of random graphs, the assumption of \emph{self-averaging} is commonly invoked when analyzing macroscopic properties such as degree distributions, centrality measures, or path statistics (cf.\ \cite{SELFroy2006small,coolen2017generating}).\footnote{The self-averaging property refers to the phenomenon whereby, in sufficiently large systems, the observable characteristics of a single realization converge to the ensemble average computed over many independent realizations. That is, a single large system behaves statistically like the average over the ensemble.}
We now derive the degree (or occupancy) distribution under a coarse-grained random graph ensemble in the sparse regime, using the self-averaging assumption.
Self-averaging is an asymptotic property stating that the empirical degree distribution observed in a single realization of the graph converges (in probability) to the ensemble-averaged distribution. Specifically, let \( N_k \) denote the number of nodes of degree \( k \) in a graph with \( N \) nodes. Then, as \( N \to \infty \), the indgree (visit) distribution can be written as:
\[
\frac{N_k}{N} \xrightarrow{P} \mathbb{E}\left[ \frac{N_k}{N}  \right]=P(k)=\sum_{\{\mathcal{A}\}}\sum_{i}\mathbb{P}_{\mathcal{C}}(\mathcal{A}) \delta_{k,k_i},
\]
where the sum runs over all possible edge configurations incident to node \( i \), and the delta function selects those resulting in degree \( k \). So this ensemble average is representative of the typical outcome in a single realization. So the occupancy distribution can be recovered as degree distribution by using the self-averaging assumption in a random graph ensemble with coarse-grained link probabilities $p_{ij}$:
\begin{remark}[Self-averaging limit]
	Under the self-averaging assumption, the degree distribution under a given number of total links $n$, is the mixture of poisson distributions:
	\begin{equation}
	P_n(k)=\sum_{i=1}^m P(k_i=k)  \qquad \text{with} \quad P(k_i=k) = \frac{e^{-\langle k \rangle_i}\langle k \rangle_i^k}{k!}
	\end{equation}
	where $\langle k \rangle_i=n\nu_i$.
\end{remark}
\begin{proof}
We follow the generating function approach of \cite[Ch.~3]{coolen2017generating}. 
The method starts from the self-averaging property of macroscopic observables, which allows one to replace averages over realizations by ensemble averages, and then uses the integral representation of the Kronecker delta to enforce degree constraints. 
A subsequent Taylor expansion of exponential terms yields the explicit form of the degree distribution. 
In our case, the compound nature of the allocation process modifies the standard outcome from a Poisson to a compound Poisson law. 
Since the derivation closely parallels the standard case and involves only lengthy but direct algebraic manipulations, we do not present the full proof here.
\end{proof}
	
		It is important to clarify the role and interpretation of {Poissonization} in random graph models, particularly in empirical applications such as human mobility or trade networks.
		
		As discussed in previous sections, {Poissonization} refers to the replacement of a binomial distribution, characterized by a fixed number of independent trials and success probability, with a Poisson distribution of the same mean. This is often done for analytical convenience: the Poisson distribution admits closed-form calculations more easily, and it is especially accurate in the asymptotic regime where the number of potential edges is large and the edge probabilities are small.
		In particular, a second order Poisson approximation of the binomial occupancy distribution  is given by  \cite{barbour1992poisson} as:
		\begin{equation}
		P_n(k)=\sum_{i=1}^N \frac{(n\nu_i)^k}{k!} e^{-n\nu_i} \left( 1 + \mathcal{O}(n\nu_i^2,\, \tfrac{k^2}{n}) \right).
		\end{equation}
		The correction term \( \mathcal{O}(n\nu_i^2,\, k^2/n) \) quantifies the relative error between the Binomial and its Poisson approximation and the approximation is especially accurate when:
		\[
		n\nu_i^2 \ll 1 \quad \text{and} \quad \frac{k^2}{n} \ll 1.
		\]
		Thus, Poissonization introduces a useful tractable approximation of binomial edge formation, especially under sparse and large-graph limits. Nevertheless, it may fail to capture higher-order dependencies inherent in real-world networks.		

		
\subsection{Continuum-grained occupancy distribution}

Let us again interpret the latent variables as intrinsic attributes of a spatial region, capturing various structural and contextual factors. In this framework, the \( n \) colored balls represent individual travelers, where the color encodes the identity or characteristics of their origin locations. The \( N \) bins correspond to destination sites, each associated with a latent attribute \( x_i \), for \( i = 1, \ldots, N \), which quantifies the location's inherent \emph{attractiveness}, that is, its capacity to fulfill needs or offer desirable resources.

We can formalize the problem of occupancy distribution as follows:
\begin{proposition}[Continued-grained occupancy]\label{prop_continuous_occu}
	Let \( \tilde{\mathcal{G}}_{N,\kappa} \) be a probability space over multi-edge random graphs with \( N \) nodes, where each node \( i \in \{1,\dots,N\} \) is associated with an independent latent type \( (x_i, y_i) \in \mathcal{X} \times \mathcal{Y} \), drawn i.i.d.\ from the product measure \( \mu_{\mathcal{X}} \otimes \mu_{\mathcal{Y}} \).
	
	Assume that the empirical graphon $	\tilde{\mathcal{A}}_N(x,y)$
	converges in distribution in \( L^1(\mathcal{X} \times \mathcal{Y}) \) to a limiting graphon \( {\kappa(x,y)} \in L^1(\mathcal{X} \times \mathcal{Y}) \).
Define the limiting expected degree function:
	\[
	\lambda(x) := {\ell} \int_{\mathcal{Y}} \kappa(x, y) \, \mu_{\mathcal{Y}}(dy)
	\]
	
	Then for every fixed \( k \in \mathbb{N}_0 \), the empirical occupancy distribution satisfies
	\[
\frac{N_k}{N}  \xrightarrow{\mathbb{P}} P(k)= \int_{\mathcal{X}} \frac{\lambda(x)^k}{k!} e^{-\lambda(x)} \, \mu_{\mathcal{X}}(dx),
	\]
	where \( N_k \) denotes the number of vertices of degree \( k \) in \( \tilde{\mathcal{G}}_{N,\kappa} \), and the limit distribution is a {mixed Poisson law} with mixing distribution \( \mu_{\mathcal{X}} \circ \lambda^{-1} \).	
 \end{proposition}
\begin{proof}
The derivation of the occupancy distribution follows \cite{bollobas2007phase,bollobas1998random}, with minor modifications to align with our notation and ensure consistency with the rest of the paper.
\end{proof}
	The theorem states that the degree of a node with latent type \( x \in \mathcal{X} \) is asymptotically distributed as a Poisson random variable with mean \( \lambda(x) \), while the distribution of types across the population is governed by \( \mu_{\mathcal{X}} \). The overall degree distribution is thus obtained by averaging over the type distribution, resulting in a mixed Poisson law.		
		The following theorem shows the equivalency between inhomogeneous random graphs and  latent-variable (fitness) networks.
\begin{theorem}\label{theo_continuumdegreee}
		Under the constraint of a fixed number of links $n$ we have that 	$$ \ell \|\widetilde{\kappa}\|_1 
		=\frac{n}{N}. $$
		So  we can write the visiting-degree distribution we have that:
		$$ \lambda (x|n)=\frac{n}{N}\frac{\int_{\mathcal{Y}} \kappa(x, y) \, \mu_{\mathcal{Y}}(dy)}{ \|\widetilde{\kappa}\|_1 }:=n\nu_x $$
	\begin{equation}
P(k|n)=	P_n(k)= \int_{\mathcal{X}} \frac{(n\nu_x)^k}{k!} e^{-n\nu_x} \, \rho(x)dx \qquad \text{where } \; \nu_x=\frac{1}{N\|\widetilde{\kappa}\|_1}\int_{\mathcal{Y}} \kappa(x, y) \, \phi(y)dy
	\end{equation}
	in the limit of infinitely large pseudo-graphs continuous when the measures $\mu$ are absolutely continuous with respect to the Lebesgue measure so that $\mu_x(dx)=\rho(x)dx$ and $\mu_y(dy)=\phi(y)dy$ and where the unit rate $\nu_x$ is the limit of infinitely large and sparse graphs.
\end{theorem}	
\begin{proof}
	The statement follows directly from Theorem~\ref{theo_graphon} and Proposition~\ref{prop_continuous_occu}. 
	Indeed, imposing the global constraint of a fixed number of links $n$ yields the scaling $
	\ell = \frac{n}{N}\frac{1}{\|\widetilde{\kappa}\|_1}
	$.
	Substituting this into the conditional rate expression derived in Proposition~\ref{prop_continuous_occu} gives $\lambda(x\mid n)=n\nu_x$, and the Poisson mixture form of $P_n(k)$ then follows immediately.
\end{proof}	

In conclusion,  this entrance rate $\nu_x$ can be viewed as a function of the bin’s attractiveness, with \( \nu_x = f(x) \), linking accessibility to underlying attractiveness.

The process of going from discrete bin-selection chance $\nu_i$  to a continuum selection rate $\nu_x$ follows  the law of large numbers and the convergence of the empirical distribution of bin attributes, enabling a mean-field approximation in the large-$N$ limit.   The probability that a location-bin  with attractiveness \( x \) receives an incoming balls (a visit) is defined as:
\begin{align}\label{eq:rate_edg_discr}
\nu_i=\nu(x_i) = \frac{\sum_{j=1}^{N} \kappa (x_i,y_j)}{\sum_{l=1}^{N} \sum_{m=1}^{N} \kappa(x_l,y_m)},
\end{align}
which corresponds to the probability that a new trip is directed to a destination with attractiveness \( x \). The numerator accounts for all possible links that could be directed to a location-bin of attractiveness \( x \), while the denominator normalizes over all potential source-destination pairs.

In the limit of a large number of locations \( N \), attractiveness can be treated as a continuous variable and the empirical distribution of $x_i$ converges (by the Glivenko-Cantelli theorem) to the true  probability density of of attractiveness $\rho(x)$. Consequently, the empirical number of bins with attractiveness in a small interval $[x,x+dx]$ is approximated by \( N(x)dx = N \rho(x)dx \). It replaces the discrete empirical distribution with a smooth deterministic density $\rho(x)$, and integrals replace sums\footnote{This transition from sums to expectations relies on the concentration of measure phenomenon: in the large $N$ limit, the empirical average of i.i.d. random variables concentrates sharply around its expected value. This justifies approximating aggregate sums by deterministic expectations in the continuum model. Using concentration of measure to justify replacing sums with expectations is equivalent to a probabilistic mean-field approximation.} so that the discrete expression eq.~\eqref{eq:rate_edg_discr} then converges to the continuum approximation:
\begin{align}\label{eq:rate_edg_cont}
\nu_x = \frac{N \int_0^{\infty} \kappa(x,y)\, \phi(y)\, dy}{N^2 \int_0^{\infty} \int_0^{\infty} \kappa(x,y)\, \rho(x)\, \phi(y)\, dx\, dy}.
\end{align}

In this formulation, the numerator represents the average number of links directed toward a bin of attractiveness \( x \), while the denominator captures the total number of links formed across the entire system. Therefore, \( \nu_x \) quantifies the relative probability that a new link is assigned to a destination of attractiveness \( x \), under the linking kernel \( \kappa(x,y) \).

\section{Load allocation problems: vacancy, coverage and overflow}\label{sec_load}
The balls-into-bins framework encompasses several fundamental allocation load problems that address distinct aspects of random distribution processes. While in the  occupancy problem we have focused on the distribution of balls across bins, alternative formulations within the same framework reveal deeper insights into allocation load problems such as congestion phenomena, and spatial coverage. Each of these problems corresponds to a distinct mechanism underlying human mobility and resource allocation in urban systems.

We begin by considering the vacancy problem, which focuses on the number of bins that remain empty after $n$ independent allocations. In mobility terms, it quantifies the number of destinations (e.g., urban zones, service points, or facilities) that remain unvisited or unused after $n$ individuals have independently selected destinations. This measure reflects the degree of underutilization and can signal infrastructural redundancy or inequities in accessibility.

Next, we consider the overflow problem, which captures congestion or saturation effects. Specifically, it measures the number of allocations that exceed the finite capacity of bins, corresponding to infrastructure or service limits at destinations. Unlike binary notions of access or inaccessibility, overflow provides a continuous and quantitative indicator of excess demand. It is particularly suited to modeling stress in urban systems, such as overcrowded transit stations, saturated hospitals, or overloaded vehicular infrastructure.

Together, these allocation-based models offer a flexible and analytically tractable framework for understanding the emergence of crowding, queuing, and systemic failure in human mobility networks and urban resource systems.

\subsection{Vacancy \& coverage problems as unserved destination-bins}
The vacancy problem is about the number of empty bins after $n$ throws and it measures how many destinations (e.g., locations, urban zones, service facilities) remain unvisited or unused after $n$ individuals make travel choices, indicating underutilized areas or infrastructural excess.
So we will investigate the number of destinations  \( E_n \) with no visit after $n$ trips, and, in addition, the coverage as  the stopping time \( T_b=\min\{n:E_n=b \} \), which indicates the (random) number of trips required until exactly $b$ destinations  remain unvisited (empty). Consequently we will estimate how many trip-balls do you need to throw, on average, until only $b$ destination-bins are still empty. 
\begin{proposition}[Coarse-grained vacancy \& coverage]\label{prop_cvc}
	Let \( n \) balls be independently thrown into \( N \) bins. Each ball chooses bin \( i \in \{1, \dots, N\} \) with fixed probability \( \nu_i \), where \( \sum_{i=1}^N \nu_i = 1 \).
	
	In the asymptotic regime assumption of large and sparse graphs, we have the following behaviors:
	
	\begin{itemize}
		\item  $E_n$ converges in distribution to a Poisson random variable	that is $	\mathbb{P}(E_n = \zeta) \to \text{Po}(\sum_i (1-\nu_i)^n)$.
		Consequently, the vacancy probability distribution is given by:
		\begin{equation}\label{vacancy_discr}
		\frac{E_n}{N} \; \xrightarrow{\mathbb{P}}\; P_n(0) = \tfrac{1}{N}\sum_{i=1}^N (1 - \nu_i)^n \approx  \tfrac{1}{N}\sum_{i=1}^N e^{-n \nu_i}
		\end{equation} 
		which is meant to be the probability that a given destination-bin is still empty after $n$ independent trips.
		
		\item The mean stopping time in the occupancy problem:
		\begin{equation}
		\mathbb{E}[T_b]= \tfrac{1}{\min_i\{\nu_i\}} \ln \tfrac{N}{b} +\Theta(N)
		\end{equation}
		which  estimates  the average number of trips needed to reach the configuration with exactly $b$ empty destination-bins, we name it coverage. The exact constant hidden inside the big theta $\Theta(N)$ depends on the variance structure of the model (e.g., the distribution of $\nu_i$).

	\end{itemize} 
\end{proposition}
\begin{proof}
By the approximation for occupancy problems as shown in \cite{sevast1973poisson,holst1977some}, the number of empty bins satisfies $E_n \xrightarrow{d} \text{Po}(m_n)$ with $m_n\approx \sum_i e^{-n\nu_i}$, giving the vacancy probability $P_n(0)\approx \tfrac{1}{N}\sum_i e^{-n\nu_i}$. For the coverage time $T_b$, defined by $E_{T_b}=b$, let $n_b$ solve $\sum_i e^{-n_b\nu_i}=b$. Then $\mathbb{E}[T_b]=n_b+\Theta(N)$, where, using the Laplace method applied as a boundary approximation, $n_b\sim \tfrac{1}{\min_i\nu_i}\log(N/b)$. Full details are given in Appendix~\ref{app_cvc}.
\end{proof}

In the proposition above, assumptions  requires the condition  that even the destination-bin with the lowest chance of being visited still has a vanishing probability of remaining empty. That is, the balls are spread out enough so that no bin is completely ignored in the limit. Moreover, the balanced conditon  assumption ensures that the expected total number of empty bins remains bounded as \( N \to \infty \), so that, while some bins may remain empty, their number does not grow without bound. Let us, finally, notice that the fine-grained configurational approach provides the same estimate of empty bins as above, where the occupancy rate is the proportion of balls being assigned to the \( i \)-th bin given by   $ \pi_i=\sum_{j}m_{ij}/NM =\nu_i$.

\begin{corollary}[classical coupon collector]
	In the case of homogeneous random graph, where we have  uniform occupancy probabilities $\nu_i=\nu_j=\frac{1}{N},\forall i,j$,  we recover the classical problem of coupon collector where the vacancy probability and the mean stopping time are:
	\begin{align*}
	P_n(0)=&\;\left(1-\frac{1}{N} \right)^n \approx \;e^{-\frac{n}{N}}\\
	\mathbb{E}[T_b]=& \,N\ln\frac{N}{b}
	\end{align*}
\end{corollary}

Using concentration of measure, sums can be replaced by their expectations under a mean-field approximation. In the continuum-grained formulation, the vacancy probability and and the expected coverage time are given by:
\begin{align} {P_n}(0) = & \int_{\mathcal{X}} \left(1 - \nu(x)\right)^n \rho(x)\, dx \\ \mathbb{E}[T_b]=n_b \quad \text{ s.t. } & \;\int_{\mathcal{X}} e^{-n_b\nu_0 x^{\alpha}}\rho(x)dx =b/N 
\end{align}
where \( \nu(x) \) denotes the allocation intensity associated with bins of type \( x \), and \( \rho(x) \) is the probability density function over bin types. This expression constitutes the continuous analogue of the discrete mixture formula presented in equation~\eqref{vacancy_discr}, and it often yields a more tractable analytical treatment.

\subsection{Overflow problem as overloaded destination-bins}
We now consider the case where bins have limited capacities. In this setting, overflow refers to the cumulative excess allocation that occurs when the number of incoming balls assigned to a bin exceeds its capacity. Specifically, for each bin, the overflow is defined as the surplus beyond its maximal admissible load. This quantity captures the extent to which allocation demands exceed structural constraints and serves as a natural indicator of saturation.
In the context of human mobility networks, overflow models the excess flow of individuals or trips directed toward a destination or transit hub beyond its physical or operational capacity. Unlike hard-capacity exclusion models (i.e. full bins), here additional trips are still routed to the destination even after capacity is reached (under certain reasonable limits), but they generate congestion externalities. For instance, in public transportation systems, overflow occurs when the number of arriving passengers at a station exceeds its service or holding capacity, inducing delays and operational inefficiencies. In road traffic networks, overflow reflects the accumulation of vehicles beyond parking or throughput capacity, manifesting as congestion. Thus, the expected overflow provides a quantitative measure of systemic pressure on the network and identifies critical nodes subject to stress under constrained infrastructure. 

In the urn-model representation of occupancy in random graphs, each destination is modeled as a bin characterized by an intrinsic attractiveness $x_i$, proportional to its resource availability, and the capacity of each destination bin is expressed as $C_i = C_0 x_i$, with the normalization constant $C_0 = NM/\sum_i x_i$, where $NM$ is the total number of visits and in the following proposition we formalize the notion of overflow and its estimate in a coarse-grained representation of random graphs.

\begin{proposition}[Coarse-grained Overflow]
	Let \( n \) be the total number of trips independently assigned to \( N \) destination-bins, where each trip is allocated to bin \( i \in \{1, \dots, N\} \) with fixed probability \( \nu_i > 0 \), such that \( \sum_{i=1}^N \nu_i = 1 \). Assume each destination \( i \) has a finite capacity \( C_i \in \mathbb{N} \). 
		Let $Y^{(n)}_i$ the number of trips received by destination-bin $i$ after $n$ total trip-balls have been thrown, and define the overflow (or residual excess) at destination-bin \( i \) as
	$
	R_i = \max(Y^{(n)}_i - C_i, 0).
	$
	
	Then, the expected total overflow is
	\[
	\mathbb{E}[R] = \sum_{i=1}^N \mathbb{E}[R_i],
	\]
	where, under asymptotic limit of large and sparse graphs, the expectation of the overflow at destination-bin \( i \) admits the approximation:
	\[
	\mathbb{E}[R_i] =\sum_{k=C_i+1}^N \left[\sum_{j=k+1}^{\infty} \frac{(n\nu_i)^j}{j!}e^{-n\nu_i} \right] 
	\]
where  the summation only considers terms where $k>C_i$, as there would be no overflow if $k\leq C_i$.
\end{proposition}
\begin{proof}
	Let us notice that a bin cannot have overflow without being full. Therefore, overflow depends on the same probabilistic behavior that governs the number of full bins but extends it to quantify the excess beyond capacity.
	Let $Y^{(n)}_i$ the number of balls assigned to bin $i$, now following a binomial distribution with parameter $\nu_i$ (the probability that any specific ball is assigned to bin $i$). Then 	
$
	P[Y^{(n)}_i=k]=\binom{n}{k}\nu_i^k(1-\nu_i)^{n-k}  
$,
	so that the expectation of overflow $R_i$ for bin $i$ can be computed as:
	\begin{align}
	\mathbb{E}[R_i]=& \sum_{k=C_i+1}^N (k-C_i) P[Y^{(n)}_i=k] = \sum_{k=C_i+1}^N  P[Y^{(n)}_i>k]
	=  \sum_{k=C_i+1}^N \left[\sum_{j=k+1}^{\infty} \frac{(n\nu_i)^j}{j!}e^{-n\nu_i} \right]  \quad \rightarrow \; 	\mathbb{E}[R]=\sum_{i=1}^N  \mathbb{E}[R_i]
	\end{align}
which is the expected total overflow.
	
\end{proof}

The fine-grained ensemble provides a natural and exact benchmark for evaluating capacity-respecting allocation under full information of the mobility structure.
Let us emphasize that the {fine-grained enumerative ensemble} introduced in Proposition~\ref{prop_fine} corresponds to a {fine-grained allocation process under full information}, particularly relevant in human mobility theory when the entire infrastructure of origin-destination pairs and their respective capacities is known a priori. In this setting, the allocation of \( n \) links (or trips) is performed by uniformly sampling without replacement from a fixed pool of \( C = \sum_i C_i \) available stubs, where each destination bin \( i \) possesses exactly \( C_i \) slots (e.g., physical or infrastructural capacity). 
As a direct consequence, overflow is not only avoided, it is {impossible}. Thus, the probability of overflow is identically zero, and the expected number of overflowed links is trivially
$
\mathbb{E}\left[ (Z_i^{(n)} - C_i)_+ \right] = 0 , \text{for all } i.
$
This regime contrasts sharply with coarse-grained probabilistic models, in which routing decisions are made under partial information and capacity violations may occur with nonzero probability. 

Let us, finally, discuss the overflow problem in continuum-grained representation of random graphs.

\begin{proposition}[Continuum-grained Overflow]
	Let \( x \in \mathbb{R}_+ \) be a latent variable with density \( \rho(x) \), such that each bin is characterized by an attribute \( x \), a capacity function \( C(x) \in \mathbb{N} \), and a selection probability \( \nu(x) \), with normalization condition
$
	\int_{\mathbb{R}_+} \nu(x) \rho(x) \, dx = 1.
$
	Under the asymptotic limit of large and sparse graph, the total expected overflow is given by:
\begin{equation}\label{eq_overflow}
	\mathbb{E}[R] \sim \int_{\mathbb{R}^+} [n\nu(x) - C(x)]_+\; \rho(x) dx,
\end{equation}
where the overflow region \( \{x : \lambda(x) > C(x)\} \) has nonzero measure as \( n \to \infty \).

\end{proposition}
\begin{proof}
	Let \( Z(x) \sim \mathrm{Poisson}(\lambda(x)) \) with \( \lambda(x) = n \nu(x) \) and capacity \( C(x) \). The overflow at bin type \( x \) is defined as
	\[
	R(x) := \max(Z(x) - C(x), 0), \quad \text{so that} \quad \mathbb{E}[R(x)] = \sum_{k=C(x)+1}^{\infty} (k - C(x)) \cdot \frac{\lambda(x)^k}{k!} e^{-\lambda(x)}.
	\]
	
	Now observe that if \( \lambda(x) \gg C(x) \), the distribution of \( Z(x) \) is sharply concentrated around its mean \( \lambda(x) \). By the law of large numbers, we have
	$	Z(x) \approx \lambda(x) \quad \text{with high probability},$
	so
$\mathbb{E}[R(x)] \approx \max(\lambda(x) - C(x), 0).$
	
	Therefore, the total expected overflow becomes
	\[
	\mathbb{E}[R] = \int \mathbb{E}[R(x)] \rho(x) dx \approx \int \max(\lambda(x) - C(x), 0) \rho(x) dx  \sim \int_{\lambda(x) > C(x)} (\lambda(x) - C(x)) \rho(x) dx,
	\]
	which completes the proof.
\end{proof}

	In human mobility congestion, {overflow} measures {how much} the flow exceeds capacity, capturing the severity of congestion, and it quantifies the excess demand, guiding infrastructure planning and flow mitigation.  It goes beyond simple {full bins} that merely indicate which destinations exceed their capacity, identifying {where} congestion occurs.

\section{Data-driven simulations}\label{sec_simul}
To demonstrate the applicability of the theoretical framework, we refer to empirical findings from \cite{vanni2024visit}, which provide structural parameters for a real-world human mobility network. Rather than re-estimating the network, we apply the model using latent variables derived from the empirical analysis.

Each location node in the network is endowed with two latent attributes that govern its role in mobility dynamics: a destination-specific variable \( x \), encoding its \emph{attractiveness}, that is, its capacity to absorb incoming trips, and an origin-specific variable \( y \), representing its \emph{trip-generation potential}, or the propensity of agents to initiate trips from that location. The variable \( x \) typically reflects structural features such as employment density, commercial activity, and transportation infrastructure, while \( y \) captures demographic or socio-economic activity related to residential patterns \footnote{
Let us observe that hidden-variable models describe networks by assigning edge probabilities through unobserved, node-specific attributes \cite{kim2018review,rastelli2016properties,hartle2021dynamic}. These latent traits are treated as random, either because they are not directly measurable (epistemic uncertainty) or because variability is irreducible (ontological randomness). Observable features such as degrees, motifs, or link distances are thus indicators of latent structure rather than direct measurements. In this view, latent variables act as generative causes that encode structural heterogeneity, explaining heavy-tailed degrees, clustering, and modular communities. They are best regarded as mathematical constructs that bridge underlying mechanisms with statistical inference.}.
In the origin-destination mobility network analyzed in~\cite{vanni2024visit}, the random graph model adopts this latent variable structure by assigning to each destination node an attractiveness variable \( x \) corresponding to its non-residential land-use intensity, and to each origin node a variable \( y \) representing residential land-use. The interaction between origins and destinations is governed by a multiplicative linking kernel of the form
$
\kappa(x, y) = x^{\alpha} y^{\beta},
$
with exponents set to \(\alpha = \tfrac{5}{6}\) and \(\beta = \tfrac{3}{4}\). The latent variables are drawn from the probability densities \(\rho(x) \sim \rho_0 x^{-\eta}\), with an estimate coefficient of $\eta=2$ for attractiveness and \(\phi(y) \sim \phi_0 e^{-\lambda y}\), with an estimate rate of $\lambda=0.2$ for trip-generation potential, representing a heavy-tailed spatial heterogeneity at destinations and an exponentially decaying profile at origins. Moreover, we adopt the notion of \emph{natural cut-off} $x_{max}$, see \cite{boguna2004cut}, as a fundamental bound on the maximum value in scale-free fat-tailed distributions, which becomes a key element in estimating the critical properties of processes defined on networks with scale-free topologies\footnote{Extreme Value Theory states that the maximum of $N$ i.i.d.\ variables drawn from a distribution with power-law tail $\rho(x) \sim x^{-\eta}$ is itself a random variable with a well-defined distribution. The cumulative distribution of the maximum is $\Pi(x) = \left( \int_1^x \rho(y)\,dy \right)^N \approx \left(1 - x^{1-\eta}\right)^N$ for large $x$, so its density is $\pi(x) = \frac{d\Pi}{dx} \sim N x^{-\eta} (1 - x^{1-\eta})^{N-1}$. The expected maximum, or natural cut-off, is then $x_c(N) = \int x\, \pi(x)\,dx$, which scales as $x_c(N) \sim N^{1/(\eta - 1)}$ for large $N$, consistently with the heuristic condition $N \int_{x_c}^\infty \rho(x)\,dx \sim 1$.}, so that $x_{max}\sim N^{\frac{1}{\gamma -1}}$.

These functional forms and parameter values will serve as the baseline specification for all subsequent Monte Carlo simulations, numerical experiments, and analytical derivations presented in this study.

\subsection*{$\rhd$  \textbf{Coarse}-grained ensemble}
We consider a canonical ensemble of directed graphs with \( N \) nodes, where the total number of links is exactly \( n \), and no hard capacity constraints are imposed on the edge multiplicities. Each node \( i \) is endowed with two latent variables: an \emph{attractiveness score} \( x_i \sim \rho(x) \propto x^{-\eta} \), drawn from a truncated Pareto distribution, and a \emph{trip-generation potential} \( y_i \sim \phi(y) = \lambda e^{-\lambda y} \), drawn from an exponential distribution.
These latent variables define a link propensity kernel, where the unnormalized weight for a directed edge from node \( j \) to node \( i \) is given by
$
p_{ij} := x_i^{\alpha} y_j^{\beta}.
$

We generate a configuration \( A = \{a_{ij}\} \in \mathbb{N}^{N \times N} \) by drawing \( n \) links independently from the categorical distribution over node pairs with probabilities \( p_{ij} \). That is, each link is assigned to a directed pair \( (i,j) \) with probability \( p_{ij} \), and the number of times each pair is selected determines the corresponding entry \( a_{ij} \) in the adjacency matrix:
\[
a_{ij} := \text{number of times } (i,j) \text{ was selected}.
\]
This repeated categorical sampling procedure is equivalent in distribution to drawing the vector of occupancies \( (a_{11}, \ldots, a_{NN}) \) from a multinomial distribution:
$
(a_{11}, \ldots, a_{NN}) \sim \mathrm{Multinomial}(n, \{\pi_{ij}\}).
$
The resulting ensemble defines the set of directed multigraphs with exactly \( n \) links and unconstrained edge multiplicities:
\[
\mathcal{G}_{N,\textbf{p}}(n) = \left\{ A \in \mathbb{N}^{N \times N} \,\middle|\, \sum_{i,j} a_{ij} = n \right\},
\]
with sampling probability governed by the latent scores. In the sparse regime (\( n \ll N^2 \)), the in-degrees \( k_i^{\text{in}} = \sum_j a_{ij} \) behave as mixtures of binomial random variables, with success probabilities modulated by the node's latent attributes. This mechanism naturally leads to broad-tailed degree distributions, often displaying approximate power-law behavior with exponents dependent on the latent variable distributions and kernel exponents \( \alpha, \beta \).

\begin{algorithm}[H]\label{code_coarse}
	\caption{Sampling of a Coarse-grained Latent Variable Graph}
	\SetKwInOut{Input}{Input}\SetKwInOut{Output}{Output}
	\Input{Number of nodes $N$, total number of links $n$, kernel exponents $\alpha$, $\beta$}
	\Output{Adjacency matrix $A \in \mathbb{N}^{N \times N}$}
	
	\BlankLine
	\textbf{1. Sample latent variables for each node:} \\
	\Indp
	Draw $x_i \sim \rho(x) \propto x^{-\eta}$ \tcp*{Attractiveness scores (destinations)} 
	Draw $y_i \sim \phi(y) = \lambda e^{-\lambda y}$ \tcp*{Trip generation scores (origins)}
	\Indm
	
	\BlankLine
	\textbf{2. Construct the link propensity matrix:} \\
	\Indp
	Compute $p_{ij} := x_i^{\alpha} \cdot y_j^{\beta}$ for all $(i,j)$ \\
	\Indm
	
	\BlankLine
	\textbf{3. Initialize adjacency matrix:} \\
	\Indp
	Set $A_{ij} \leftarrow 0$ for all $(i,j)$
	\Indm
	
	\BlankLine
	\textbf{4. Repeat categorical sampling:} \\
	\Indp
	\For{$k = 1$ \KwTo $n$}{
		Sample index $(i,j)$ with probability $p_{ij}$ \\
		Set $A_{ij} \leftarrow A_{ij} + 1$
	}
	\Indm
	
	\BlankLine
	\Return $A$
\end{algorithm}

\subsection*{$\rhd$ \textbf{Fine}-grained ensemble}
We consider the ensemble of directed multigraphs with $N$ nodes, where the total number of links is exactly $n$ and the maximum number of links allowed between each ordered node pair $(i,j)$ is constrained by a non-negative integer capacity matrix $m = \{m_{ij}\}$. The ensemble is defined as the set of all adjacency matrices $A = \{a_{ij}\} \in \mathbb{N}^{N \times N}$ such that each $a_{ij} \le m_{ij}$ and the total number of links is conserved:
\begin{equation*}
\mathcal{G}_{N,\mathbf{m}}(n) = \left\{ A \in \mathbb{N}^{N \times N} \,\middle|\, 0 \le a_{ij} \le m_{ij},\; \sum_{i,j} a_{ij} = n \right\}.
\end{equation*}

The probability measure on this set assigns \emph{equal probability} to every admissible matrix $A$, i.e., the uniform measure over $\mathcal{G}_{N,\mathbf{m}}(n)$. This uniform allocation is the defining property of the microcanonical (or fine-canonical) ensemble.

As described in pseudo-code \ref{code_fine}, we assume that the capacity matrix $\mathbf{m}$ is generated by an underlying latent-variable model. Specifically, each node $i$ is assigned a latent \emph{attractiveness} score $x_i$ drawn from a truncated Pareto distribution $\rho(x) \propto x^{-\eta}$ and a latent \emph{trip-generation potential} $y_i$ drawn from an exponential distribution $\phi(y) = \lambda e^{-\lambda y}$. The raw propensity for a directed link from node $j$ to node $i$ is then given by
\[
p_{ij} := x_i^{\alpha} y_j^{\beta},
\]
and the normalized edge probabilities are defined as $\pi_{ij} = p_{ij} / \sum_{i,j} p_{ij}$. The slot capacity matrix is then generated as a multinomial random vector since the normalized values $\{\pi_{ij}\}$ specify the probabilities that a capacity slot is associated with the node pair $(i,j)$. A total of $NM$ slots are drawn:
\[
(m_{11}, \dots, m_{NN}) \sim \mathrm{Multinomial}(NM, \{\pi_{ij}\}),
\]
thereby forming a latent infrastructure of link capacities $m_{ij}$ with $\sum_{i,j} m_{ij} = NM$, and each $m_{ij} \sim \mathrm{Binomial}(NM, \pi_{ij})$ marginally.

Given a realization of $\mathbf{m}$, a matrix $A$ is then sampled uniformly from $\mathcal{G}_{N,\mathbf{m}}(n)$ using a two-stage \emph{nested multivariate hypergeometric allocation}. First, the row totals $r_i$ are assigned by scalar hypergeometric draws according to the row capacity sums $R_i = \sum_j m_{ij}$, under the constraint $\sum_i r_i = n$. Second, for each row, the allocated row quota $r_i$ is distributed among the columns using further hypergeometric draws, ensuring that $0 \le a_{ij} \le m_{ij}$ and $\sum_j a_{ij} = r_i$. This sequential procedure guarantees exact satisfaction of both the global link constraint and all local capacity constraints, and produces $A$ with uniform probability over $\mathcal{G}_{N,\mathbf{m}}(n)$.

In the sparse regime ($n \ll NM$), the in-degrees $k_i^{\mathrm{in}} = \sum_j a_{ij}$ and out-degrees $k_j^{\mathrm{out}} = \sum_i a_{ij}$ behave as mixtures of Binomial random variables parameterized by the latent scores. The model thus provides a fine-grained (microcanonical) random graph with a latent-variable structure, strictly enforcing all hard constraints and yielding exact combinatorial statistics.

\begin{algorithm}[H]\label{code_fine}
	\caption{Sampling of a Fine-grained Latent Variable Graph}
	\SetKwInOut{Input}{Input}
	\SetKwInOut{Output}{Output}
	\Input{Number of nodes $N$; total number of links $n$; total number of slots $NM$}
	\Output{Adjacency matrix $A = \{a_{ij}\}$ such that $a_{ij} \le m_{ij}$ and $\sum_{i,j} a_{ij} = n$}
	
	\BlankLine
	\textbf{Step 1: Sample latent variables for each node}\;
	\For{$i \gets 1$ \KwTo $N$}{
		$x_i \gets$ draw from truncated Pareto distribution $\rho(x) \propto x^{-\eta}$\;
		$y_i \gets$ draw from exponential distribution $\phi(y) = \lambda e^{-\lambda y}$\;
	}
	
	\BlankLine
	\textbf{Step 2: Compute and normalize link propensities}\;
	\For{$i, j \in \{1,\ldots,N\}$}{
		$p_{ij} \gets x_i^\alpha \, y_j^\beta$\;
	}
	$\pi_{ij} \gets p_{ij} / \sum_{u,v} p_{uv}$\;
	
	\BlankLine
	\textbf{Step 3: Generate slot capacity matrix}\;
	$(m_{11}, \ldots, m_{NN}) \gets \mathrm{Multinomial}(NM, \{\pi_{ij}\})$\;
	Let $m_{ij}$ be the number of slots for edge $(i,j)$\;
	
	\BlankLine
	\textbf{Step 4: Sample uniformly from all admissible allocations with total $n$ and local caps $m_{ij}$}\;
	Sample $A$ \emph{uniformly at random} from $\mathcal{G}_{N,\mathbf{m}}(n)$ as follows:\;
	
	\BlankLine
	\Indp
	\textbf{(a) Allocate row totals via hypergeometric draws}\;
	Let $R_i = \sum_{j=1}^{N} m_{ij}$\;
	Let $n_{\mathrm{rem}} \gets n$, $M_{\mathrm{rem}} \gets NM$\;
	\For{$i \gets 1$ \KwTo $N-1$}{
		$r_i \gets$ draw from $\mathrm{Hypergeometric}(R_i, M_{\mathrm{rem}} - R_i, n_{\mathrm{rem}})$\;
		$n_{\mathrm{rem}} \gets n_{\mathrm{rem}} - r_i$\;
		$M_{\mathrm{rem}} \gets M_{\mathrm{rem}} - R_i$\;
	}
	$r_N \gets n_{\mathrm{rem}}$\;
	
	\BlankLine
	\textbf{(b) For each row, allocate $r_i$ links among columns via hypergeometric draws}\;
	\For{$i \gets 1$ \KwTo $N$}{
		$M_i \gets R_i$\;
		$n_i \gets r_i$\;
		\For{$j \gets 1$ \KwTo $N-1$}{
			$a_{ij} \gets$ draw from $\mathrm{Hypergeometric}(m_{ij}, M_i - m_{ij}, n_i)$\;
			$n_i \gets n_i - a_{ij}$\;
			$M_i \gets M_i - m_{ij}$;
		}
		$a_{iN} \gets n_i$\;
	}
	\Indm
	
	\BlankLine
	\Return{$A = \{a_{ij}\}$}
\end{algorithm}

The coarse-grained ensemble generation algorithm (Algorithm~\ref{code_coarse}) relies on repeated categorical sampling and achieves an overall time complexity of $\mathcal{O}(N^2 + n)$. This method is highly efficient, particularly in the sparse regime where the total number of links satisfies $n \ll N^2$. In contrast, the fine-grained ensemble (Algorithm~\ref{code_fine}) employs a nested multivariate hypergeometric sampling procedure, resulting in a higher computational cost of $\mathcal{O}(N^2 + NM)$, where $NM$ is the total number of available link slots. Since $NM \gg n$ in typical sparse settings, this additional cost reflects the fine-grained model's strict enforcement of local capacity constraints. Consequently, while the coarse-grained ensemble is computationally advantageous and well-suited to large-scale network models, the fine-grained approach provides exact control over edge multiplicities and guarantees uniform sampling over the admissible configuration space, at the expense of increased computational overhead.
An alternative to the fine-grained approach is offered by the latent stub-matching ensemble. Although a full treatment of this method lies beyond the main scope of the paper, for completeness we provide a detailed description in Appendix~\ref{app_stub}. The procedure begins by generating a random capacity matrix from latent variables, as described above, and then allocating the required number of links by sampling uniformly from the resulting pool of available stubs. This construction defines an ensemble that is both efficient and analytically tractable: it approximates the fine-grained model in the sparse regime while retaining much of the computational speed of the coarse-grained approach. As a result, the latent stub-matching method is particularly well-suited for large, heterogeneous networks with latent structure.

\subsection*{$\rhd$ \textbf{Continuum}-grained ensemble}
The principal motivation for introducing the continuum-grained ensemble is to establish a general framework that enables the derivation of analytical expressions for classical urn models, such as occupancy, vacancy, and congestion, in random allocation problems. In this ensemble, each bin or node is endowed with a latent variable and the assignment of balls (or links) is governed by a kernel function defined on these latent variables. By considering the large-system limit in which the discrete set of bins is replaced by a continuum, the model admits integral representations for key probabilistic quantities. This approach unifies and extends traditional balls-in-bins analyses, allowing the explicit computation of quantities such as the probability of a bin being unoccupied (vacancy), the expected load or congestion, and other statistics, directly in terms of the latent variable distribution and the kernel. The continuum-grained perspective thus serves as a powerful analytical tool for studying random allocation processes in large heterogeneous systems. 
In the following, we present the analytical results obtained within the continuum-grained ensemble framework, which serves as the continuous counterpart to the discrete coarse-grained and fine-grained ensembles used in the Monte Carlo simulations. Table~\ref{tab_param} summarizes the key normalization constants employed in the subsequent calculations in the case  \(X \sim \rho(x) = \rho_0 x^{-\eta} \mathbf{1}_{[a,x_c)}\) with \(\eta > 1\), and \(\nu(x) = \nu_0 x^{\alpha}\) for \(0<\alpha <1\) as discussed in Theorem~\ref{theo_continuumdegreee} and explicitly expressed in eq.\eqref{eq:rate_edg_cont}.
\begin{table}[!ht]
	\begin{center}
		\scalebox{0.8}{%
			\renewcommand{\arraystretch}{2.75}
			\begin{tabular}{
					>{\centering\arraybackslash}m{4.5cm}  
					>{\columncolor{gray!10}}>{\centering\arraybackslash}m{6.5cm}  
					>{\centering\arraybackslash}m{4.5cm}                         
					@{}
				}
				\toprule
				\shortstack{ \textbf{Normalization} \\ {Constants}} & 	\shortstack{ \textbf{General Expression} \\ $\rho(x)\sim x^{-\eta}$ and $\kappa (x,y)= x^{\alpha} y^{\beta}$} & \shortstack{\textbf{Special Case:} $N\to \infty$\\ $a = 1$, $\eta = 2$, $\alpha = \tfrac{5}{6}$} \\
				\midrule		
				
				\shortstack{	Maximum size: \\ $x_{\max}$} &
				$ N^{\tfrac{1}{\eta - 1}}$ &
				$N$ \\
				
				\shortstack{Density normalization:\\ $\rho_0$} &
				$\displaystyle \frac{(\eta - 1)\, a^{\eta - 1}}{1 - \left( \frac{a}{N^{\frac{1}{\eta - 1}}} \right)^{\eta - 1}}$ &
				$1$ \\
				
				\shortstack{Expected value:\\ $\mathbb{E}[x]$} &
				$\displaystyle \frac{\eta - 1}{a^{\eta - 1} - N^{1 - \eta}} \left( N^{\tfrac{2 - \eta}{\eta - 1}} - a^{2 - \eta} \right)$ &
				$\frac{Na}{N-a}\log\frac{N}{a}=\log N$ \\
				
				\shortstack{ Fractional moment: \\ $\mathbb{E}[x^\alpha]$ } &
				$\displaystyle \frac{\eta - 1}{a^{\eta - 1} - N^{1 - \eta}} \cdot \frac{N^{\tfrac{\alpha - \eta + 1}{\eta - 1}} - a^{\alpha - \eta + 1}}{\alpha - \eta + 1}$ &
				$\frac{1}{1-\alpha}=\displaystyle 6 $ \\
				
				\shortstack{Visitation rate:  \\ $\nu_0$ s.t. $\nu_x=\nu_0 x^{\alpha}$} &
				$\nu_0=\displaystyle \frac{1}{N\, \mathbb{E}[x^\alpha]}$ &
				$\nu_0=\displaystyle \frac{1}{6N}$ \\
				
				\shortstack{ Balls proportion: \\ $n_0$ s.t. $n(x)=n_0x^{\alpha}$ } &
				$n_0=\displaystyle \frac{1}{n N\, \mathbb{E}[x^\alpha]}$ &
				$n_0=\displaystyle \frac{1}{6nN}$ \\
				
				\shortstack{Empty-bins proportion: \\ $b_0$ s.t. $b(x)=b_0x ^{\alpha}$} &
				$b_0=\displaystyle \frac{b}{N\, \mathbb{E}[x^\alpha]}$ &
				$b_0=\displaystyle \frac{b}{6N}$ \\
				
				\shortstack{ Overflow proportion:\\  $C_0$ s.t. $C(x)=C_0 x$ }&
				$C_0=\displaystyle \frac{MN}{N\, \mathbb{E}[x]}$ &
				$C_0=\displaystyle \frac{M}{\log N}$ \\
				
				\bottomrule
			\end{tabular}
		}
	\end{center}
	\caption{Analytical evaluation in the continuum approximation of the important normalization constants in the random graph models. Parameters which will be held constant are $N=10^3$ and $M=10^3$ (i.e. $\mathfrak{n}=1$). }
	\label{tab_param}
\end{table}

\paragraph{$\bullet$ Occupancy problem}
In the parametrization considered here and from the conclusions in Theorem \ref{theo_continuumdegreee} we can write the visiting occupancy distribution (in-degree distribution) is asymptotically, for large and sparse graph,  the following (mixed-Poisson) degree distribution:
\begin{align*}
P_{n}(k)
&=\frac{n^{k}}{k!}\int_{0}^{\infty}\mathrm e^{-n\,\nu_{0}x^{\alpha}}\, \bigl(\nu_{0}x^{\alpha}\bigr)^{k}\, \rho_{0}\,x^{-\eta}\,\mathrm dx 
=\frac{\rho_{0}\,\bigl(n\nu_{0}\bigr)^{k}}{k!}
\int_{0}^{\infty}
x^{k\alpha-\eta}\,
\mathrm e^{-n\nu_{0}x^{\alpha}}\,\mathrm dx\\
&=\frac{\rho_{0}\,(n\nu_{0})^{k}}{k!}\,
\frac{1}{\alpha}\,(n\nu_{0})^{-\frac{k\alpha-\eta+1}{\alpha}}
\int_{0}^{\infty}
u^{\frac{k\alpha-\eta+1}{\alpha}-1}\,
\mathrm e^{-u}\,\mathrm du  \\
&=\frac{\rho_{0}}{\alpha}\,(n\nu_{0})^{\frac{\eta-1}{\alpha}}\,
\frac{\Gamma\!\bigl(k-\tfrac{\eta-1}{\alpha}\bigr)}
{\Gamma(k+1)} 
\end{align*}

where, in the second line, we have set
\(
u = n\nu_{0}x^{\alpha}
\;\Longrightarrow\;
x=(u/n\nu_{0})^{1/\alpha},\;
\) with $\mathrm dx
=\tfrac{1}{\alpha}\,
(n\nu_{0})^{-\frac{1}{\alpha}}\,
u^{\frac{1}{\alpha}-1}\mathrm du$
,
and the consequent integral at origin converges when  \(k>\tfrac{\eta-1}{\alpha}\).


Using Stirling’s expansion
\(
\displaystyle
\tfrac{\Gamma(k-a)}{\Gamma(k+1)}
\simeq k^{-(a+1)}\bigl[1+\mathcal O(k^{-1})\bigr]
\)
with \(a=(\eta-1)/\alpha\), we obtain the occupancy (in-degree) distribution as:
\begin{equation}\label{eq_occup}
	P_{n}(k)\;\sim\;
	\tfrac{\rho_{0}}{\alpha}\,
	(n\nu_{0})^{\frac{\eta-1}{\alpha}}\,
	k^{-\mu}.
\end{equation}
Hence, in the sparse regime ($n=o(N^{2})$) exhibiting a power-law tail with exponent  
\[
\mu \;=\; 1+\frac{\eta-1}{\alpha},
\]
while the prefactor carries the explicit $n$-dependence
$n^{(\eta-1)/\alpha}$. In  the particular case of simulations with parameters $\eta=2$ and $\beta=5/6$ we have that  $\mu= 1+ \frac{2-1}{5/6}=1+\frac{6}{5}\approx 2.2$.

Fig.~\ref{occupancies} presents the results of Monte Carlo simulations based on both fine-grained and coarse-grained generative procedures, alongside the analytical prediction obtained from the continuum-grained framework.
\begin{figure}[!ht]
	\centering
	\begin{subfigure}[c]{0.9\textwidth}
		\centering
		\includegraphics[width=0.7\linewidth]{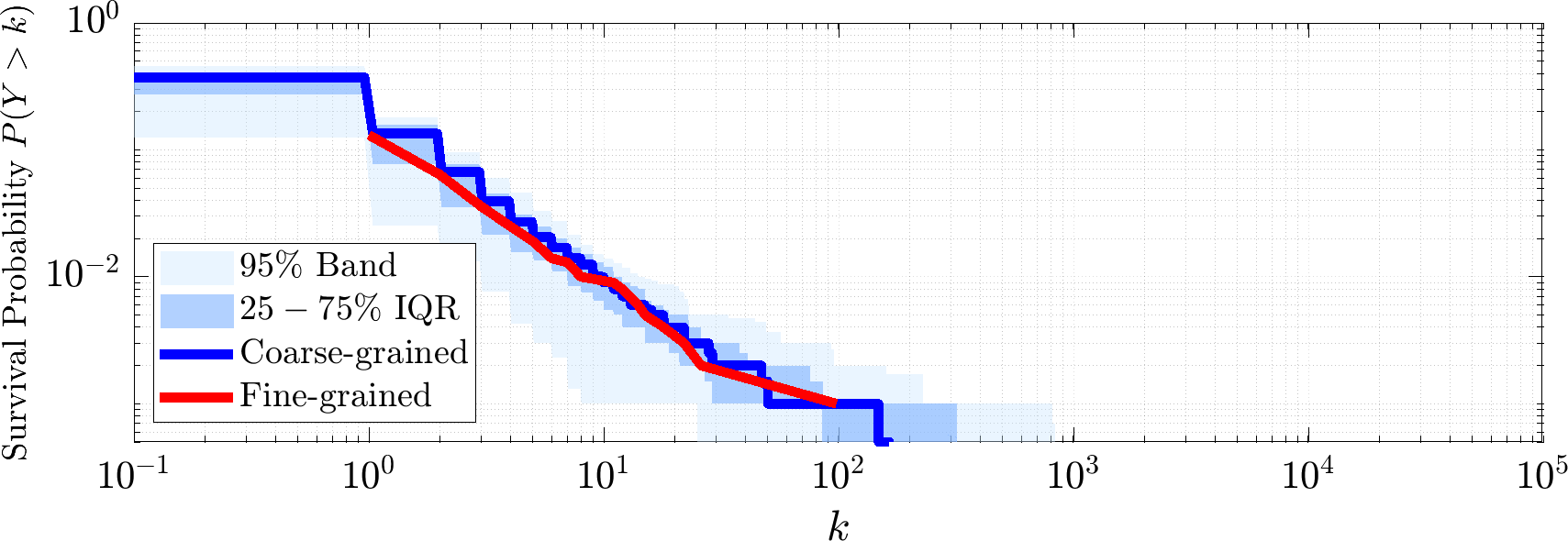}
		\caption{$n=10^3$}
	\end{subfigure} \\
	\vspace{0.5cm}
	\centering
	\begin{subfigure}[c]{0.9\textwidth}
		\centering
		\includegraphics[width=0.7\linewidth]{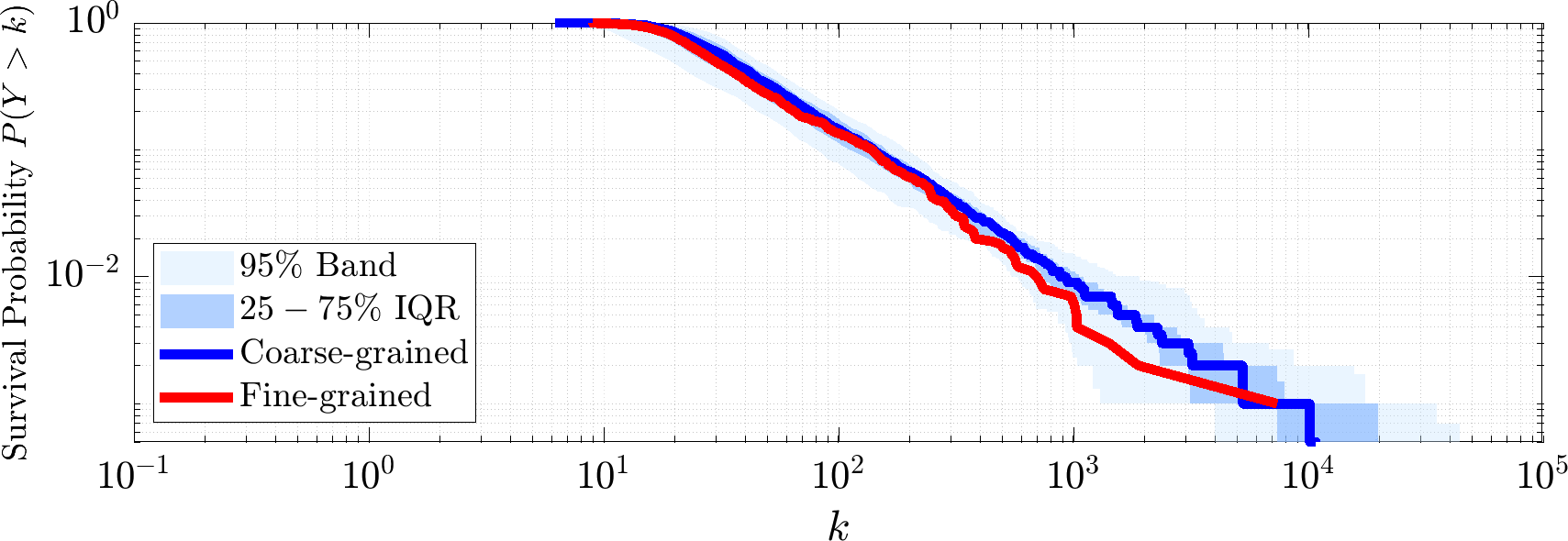}
		\caption{$n=10^5$}
	\end{subfigure}\\
\vspace{0.5cm}
\centering
\begin{subfigure}[c]{0.9\textwidth}
	\centering
	\includegraphics[width=0.7\linewidth]{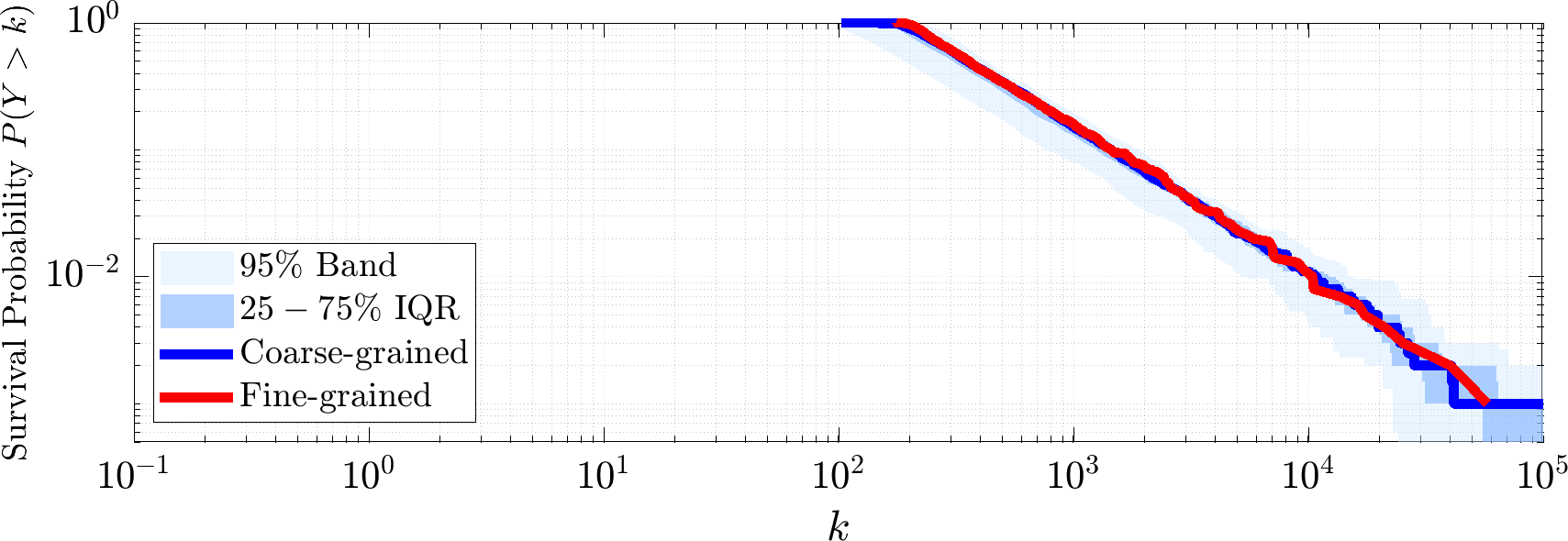}
	\caption{$n=10^6$}
\end{subfigure}
	\caption{Occupancy problem interpreted as  in-degree distribution of new visits received by destination bins, under varying numbers of ball draws \( n \) (i.e., total trips). The plot reports results from $S=50$ independent realizations using both fine-grained and coarse-grained Monte Carlo generative models. The observed survival distribution exhibits a fat-tailed profile, consistent with a power-law probability density function with exponent \( \mu = 2.2 \), in agreement with the predictions of the continuum approximation in eq.\eqref{eq_occup}.}
	 \label{occupancies}
\end{figure}

\paragraph{$\bullet$ Vacancy problem}
The vacancy problem is about the number of empty bins after \( n \) throws and it measures how many destinations (e.g., locations, urban zones, service facilities) remain unvisited or unused after \( n \) individuals make travel choices, indicating underutilized areas or infrastructural excess. In particular we will investigate the number of destinations \( E_n \) with no visit after \( n \) trips. Let \( x \in [a, x_M] \) denote the latent attractiveness of a destination, drawn from a fat-tail distribution \( \rho(x) \sim \rho_0 x^{-\eta} \), with \( \eta > 1 \) and support \( [a, x_M] \subset \mathbb{R}_+ \), where \( a > 0 \). The probability that a destination with attractiveness \( x \) receives a trip in a single throw is given by the kernel \( \nu(x) = \nu_0 x^\alpha \), with \( \nu_0 > 0 \) and \( \alpha > 0 \). Assuming independent and identically distributed throws, the probability that destination \( x \) receives no visits after \( n \) trips is \( p_x(0 \mid n) = (1 - \nu_0 x^\alpha)^n \). Therefore, the expected fraction of vacant destinations is given by the integral \( P(0, n) = \int_a^{x_M} (1 - \nu_0 x^\alpha)^n \rho_0 x^{-\eta} dx \). Setting \( \beta = -\nu_0 \), we rewrite this as a binomial differential \( I_\alpha := \int_a^{x_M} (1 + \beta x^\alpha)^n x^{-\eta} dx \). For integer \( n \), the antiderivative is expressible in terms of the incomplete Beta function \cite{chaudhry1994generalized} as \( P(0, n) = \frac{\rho_0}{\alpha} (-\nu_0)^{\frac{1 - \eta}{\alpha}} \left[ B(-\nu_0 x^\alpha; \tfrac{1 - \eta}{\alpha}, n) \right]_{x = a}^{x = x_M} \). Considering \( x_M < \infty \) as the natural cutoff $N^{1/(\eta-1)}$, the integral is finite and we may investigate the asymptotic regime as \( n \to \infty \). Define \( s = \frac{1 - \eta}{\alpha} \), and write \( P(0, n) \approx \frac{\rho_0}{\alpha} (n \nu_0)^{-s} \Gamma(s, n \nu_0 a^\alpha) \), where \( \Gamma(s, x) \) is the upper incomplete Gamma function. For large \( n \), the asymptotic expansion \( \Gamma(s, x) \sim x^{s - 1} e^{-x} \) yields to the vacancy distribution of visits as:
\begin{align}
 P(0, n)=&\tfrac{\rho_0}{\alpha\, \nu_0^{\frac{1 - \eta}{\alpha}}}
 \left[
 B_{}
 \left( \nu_0\,N^{\frac{\alpha}{\eta - 1}}; \tfrac{1 - \eta}{\alpha},\, n + 1 \right)
 -
 B_{}
 \left(\nu_0 a^\alpha; \tfrac{1 - \eta}{\alpha},\, n + 1 \right)
 \right],
 \\  \sim & \; \tfrac{\rho_0a^{1-\eta}}{\alpha \nu_0 a^{\alpha}  }  \cdot \frac{\exp\{-(\nu_0 a^\alpha) n \} }{n} \qquad \text{as }\; n\to \infty \label{eq_asy_empty}
\end{align}
This shows that the expected fraction of vacant destinations decays exponentially fast with a power-law prefactor. If we analogously model incoming trip probabilities, and assume independence between origin and destination effects, the expected fraction of destinations that are fully disconnected (neither origins nor destinations of any trip) becomes \( P(0, n) = P_{\text{out}}(0 \mid n) \cdot P_{\text{in}}(0 \mid n) \), with each term admitting similar asymptotic scaling.
Fig.~\ref{fig_empty} presents the results of Monte Carlo simulations based on both fine-grained and coarse-grained generative procedures, alongside the analytical prediction obtained from the continuum-grained framework.
\begin{figure}[!ht]
	\centering
	\begin{subfigure}[c]{0.9\textwidth}
		\centering
		\includegraphics[width=0.7\linewidth]{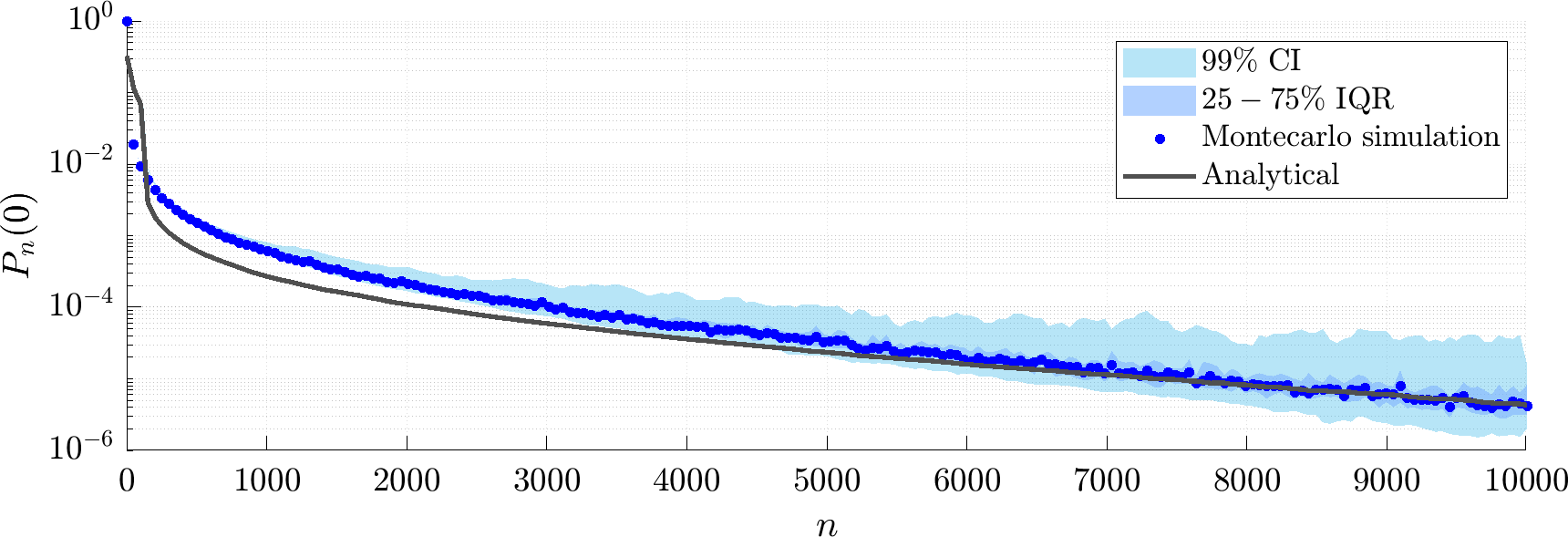}
		\caption{Coarse-grained}
	\end{subfigure} \\
	\vspace{0.65cm}
	\centering
	\begin{subfigure}[c]{0.9\textwidth}
		\centering
		\includegraphics[width=0.7\linewidth]{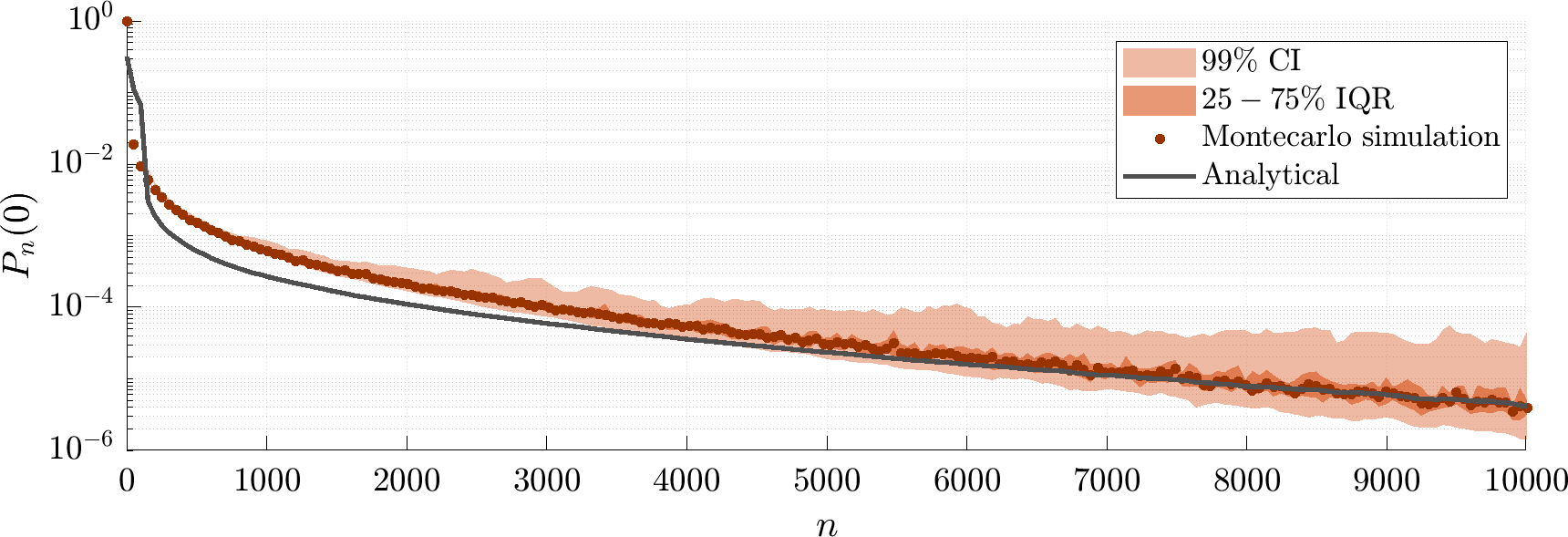}
		\caption{Fine-grained}
	\end{subfigure}
\caption{Vacancy problem as the expected fraction of empty destination bins after \( n = 10^5 \) draws, computed for (a) the coarse-grained ensemble graph and (b) the fine-grained ensemble graph, each averaged over \( S = 50 \) independent realizations. Both results show good agreement with the asymptotic analytical expression derived from the continuum-grained framework, as given in Eq.~\eqref{eq_asy_empty}.} \label{fig_empty}
\end{figure}

\paragraph{$\bullet$ Coverage problem}
 Assuming the concentration assumption, we approximate the stopping-time $\sum_{i=1}^N e^{-n_b \nu_0 x_i^\alpha}$ as $N \int_a^\infty \rho_0 x^{-\eta} e^{-n_b \nu_0 x^\alpha} \mathrm{d}x$.
 Via the change of variable \(u = n_b \nu_0 x^{\alpha}\), one obtains \(\int_a^\infty \rho_0 x^{-\eta} e^{-n_b \nu_0 x^\alpha} \mathrm{d}x = \tfrac{\rho_0}{\alpha} (n_b \nu_0)^s \Gamma(-s, n_b \nu_0 a^{\alpha})\), where \(s = \tfrac{\eta - 1}{\alpha} > 0\). For \(n_b \nu_0 a^\alpha \gg 1\), we use the asymptotic expansion \(\Gamma(\beta, z) \sim z^{\beta - 1} e^{-z}\), yielding \(\int_a^\infty \rho_0 x^{-\eta} e^{-n_b \nu_0 x^\alpha} \mathrm{d}x \sim \tfrac{\rho_0}{\alpha} a^{-(\eta - 1)} (n_b \nu_0)^{-1} e^{-n_b \nu_0 a^\alpha}\). Imposing \(N \int_a^\infty \cdots = b\), and setting \(z = n_b \nu_0 a^\alpha\), we obtain \(z e^z = \tfrac{N \rho_0}{\alpha b} a^{-(\eta - 1)}\), whose solution is given by the Lambert function $W(\cdot)$. Therefore, under the regime \(n_b \nu_0 a^\alpha \gg 1\), we have
\begin{equation}
{n_b = \frac{1}{\nu_0 a^\alpha} \, W\!\left( \frac{N \rho_0}{\alpha b} a^{-(\eta - 1)} \right)}.
\end{equation}
with a relative error of this approximation as \( \mathcal{O}\!\bigl(\tfrac{1}{n_b\nu_{0} a^{\alpha}} \bigr)\).

In our particular choice of the parameters $\eta=2$ and $\alpha=5/6$, we expect to have for $b\ll N $:
\begin{equation}\label{eq_exstop}
\mathbb{E}[T_b] = 6N (\log \tfrac{6N}{5b} -\log \log \tfrac{6N}{5b}) + \mathcal{O}(N)
\end{equation}
where we have used the asymptotic expansion in the case of large argument for the Lambert function.
Fig.~\ref{fig_stopping} presents the results of Monte Carlo simulations based on coarse-grained generative procedures, alongside the analytical prediction obtained from the continuum-grained framework.
\begin{figure}[!ht]
	\centering
	\begin{subfigure}[c]{0.9\textwidth}
		\centering
		\includegraphics[width=0.75\linewidth]{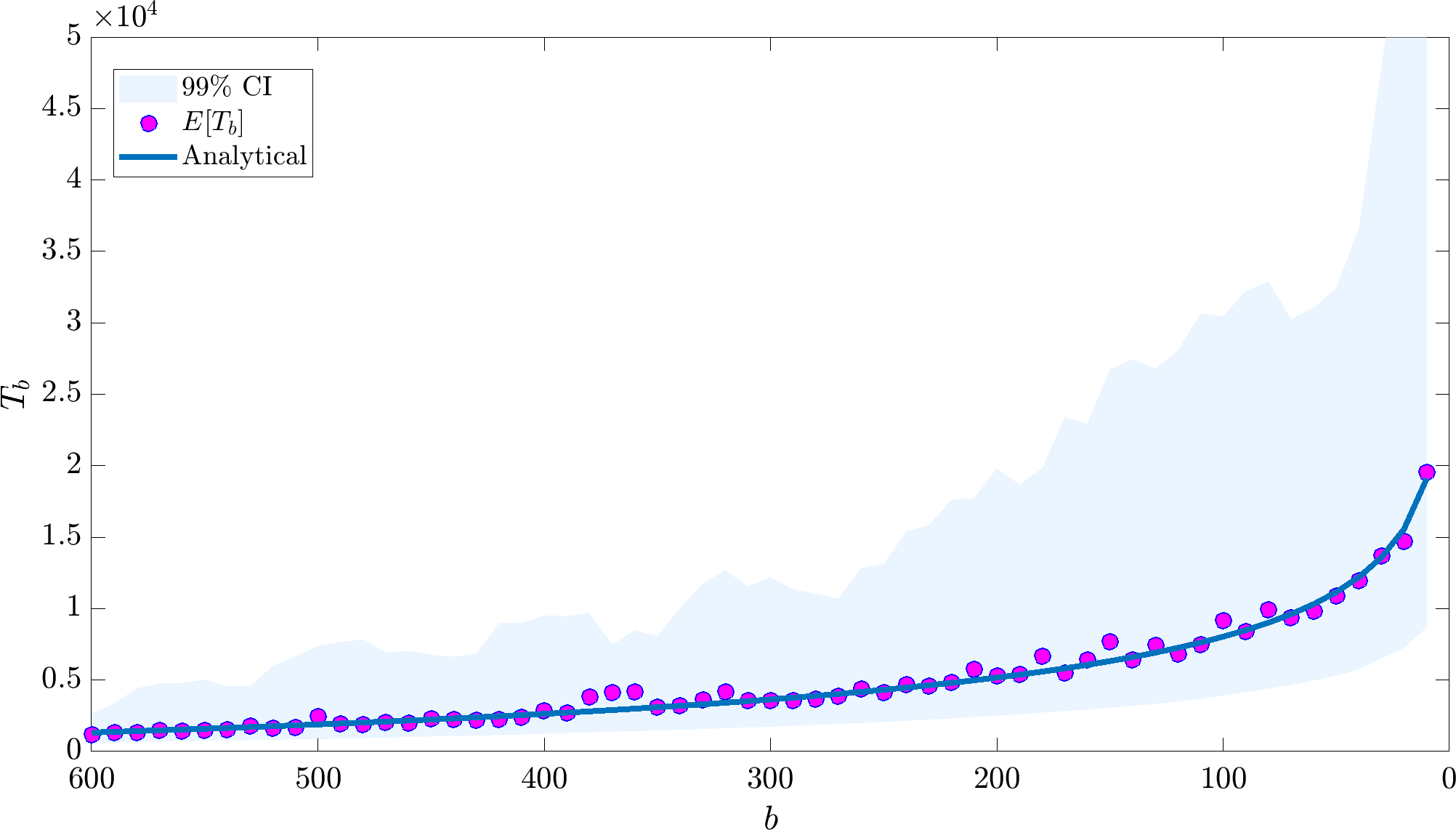}
	\end{subfigure}
		\caption{Coverage problem as the number of balls drawn before observing \( b \) empty bins out of a total of \( N = 10^3 \) bins. The results are consistent with the analytical prediction given by the closed-form expression in Eq.~\eqref{eq_exstop}.} \label{fig_stopping}
\end{figure}

\paragraph{$\bullet$ Overflow problem}
In overflow, the  capacity for each location bin becomes crucial via the capacity function as \(C(x)=C_{0}x\) where now \( x \) is the attractiveness as continuous random variable\footnote{{The constant \( C_0 \) is fixed by imposing \( \sum_{i=1}^N C(x_i) = NM \), approximated as \( C_0 N \mathbb{E}[X] = NM \).}
}.

The expected overflow as, in eq.~\eqref{eq_overflow}, is given by:
\[
\mathbb{E}[R] = N\int_a^{x_\star(n)} \left[ n\nu_0 x^\alpha - C_0 x \right] \rho_0 x^{-\eta} \, dx,
\]
where the threshold is defined by solving \( n\nu_0 x^\alpha = C_0 x \), yielding 
\begin{equation}
 x_\star(n) = \left( \frac{n\nu_0}{C_0} \right)^{\frac{1}{1-\alpha}}=\left(\frac{n}{NM}\frac{\mathbb{E}[x]}{\mathbb{E}[x^{\alpha}]}  \right)^{\frac{1}{1-\alpha}}.
 \end{equation}
Since \( \alpha < 1 \), we have \( x_\star(n) \to \infty \) as \( n \to \infty \), and overflow occurs only on the left tail \( [a, x_\star(n)] \), provided \( x_\star(n) > a \).

The expected overflow, for \( \eta \ne 2 \) and \( \eta \ne \alpha + 1 \) is:
\[
\mathbb{E}[R]=
N\rho_{0}
\Bigl[
\frac{n\nu_{0}\,a^{\alpha-\eta+1}}{\eta-\alpha-1}
+\frac{C_{0}\,a^{\,2-\eta}}{\eta-2}
\Bigr]
-
N\rho_{0}
\Bigl[
\frac{n\nu_{0}\,x_\star^{\,\alpha-\eta+1}}{\eta-\alpha-1}
-\frac{C_{0}\,x_\star^{\,2-\eta}}{\eta-2}
\Bigr] .
\]
In the case study of interest of \( \eta = 2 \), we instead find:
\begin{equation}
\mathbb{E}[R]=
N\rho_{0}
\Bigl[
\frac{n\nu_{0}a^{\alpha-1}}{1-\alpha}
-\frac{C_{0}}{1-\alpha}\,\log n
\Bigr]
+
N\rho_{0}
\Bigl[
C_{0}\log a
-\frac{C_{0}}{1-\alpha}\,
\log\!\Bigl(\tfrac{\nu_{0}}{C_{0}}\Bigr)
\Bigr].
\end{equation}
again with linear leading behavior and a subleading negative logarithmic correction.

Moreover since $\alpha=5/6$ we get the asymptotic behavior as:
\begin{equation}\label{eq_overflow_C}
E[R]\sim n(1+O(\tfrac{n}{\log n})) \qquad \text{ for } \; n\gg \tfrac{6MN}{\log N}
\end{equation}
Figure~\ref{fig_overflow} presents the results of Monte Carlo simulations based on coarse-grained generative procedures, alongside the analytical prediction obtained from the continuum-grained framework.
\begin{figure}[!ht]
	\centering
	\begin{subfigure}[c]{0.9\textwidth}
		\centering
		\includegraphics[width=0.75\linewidth]{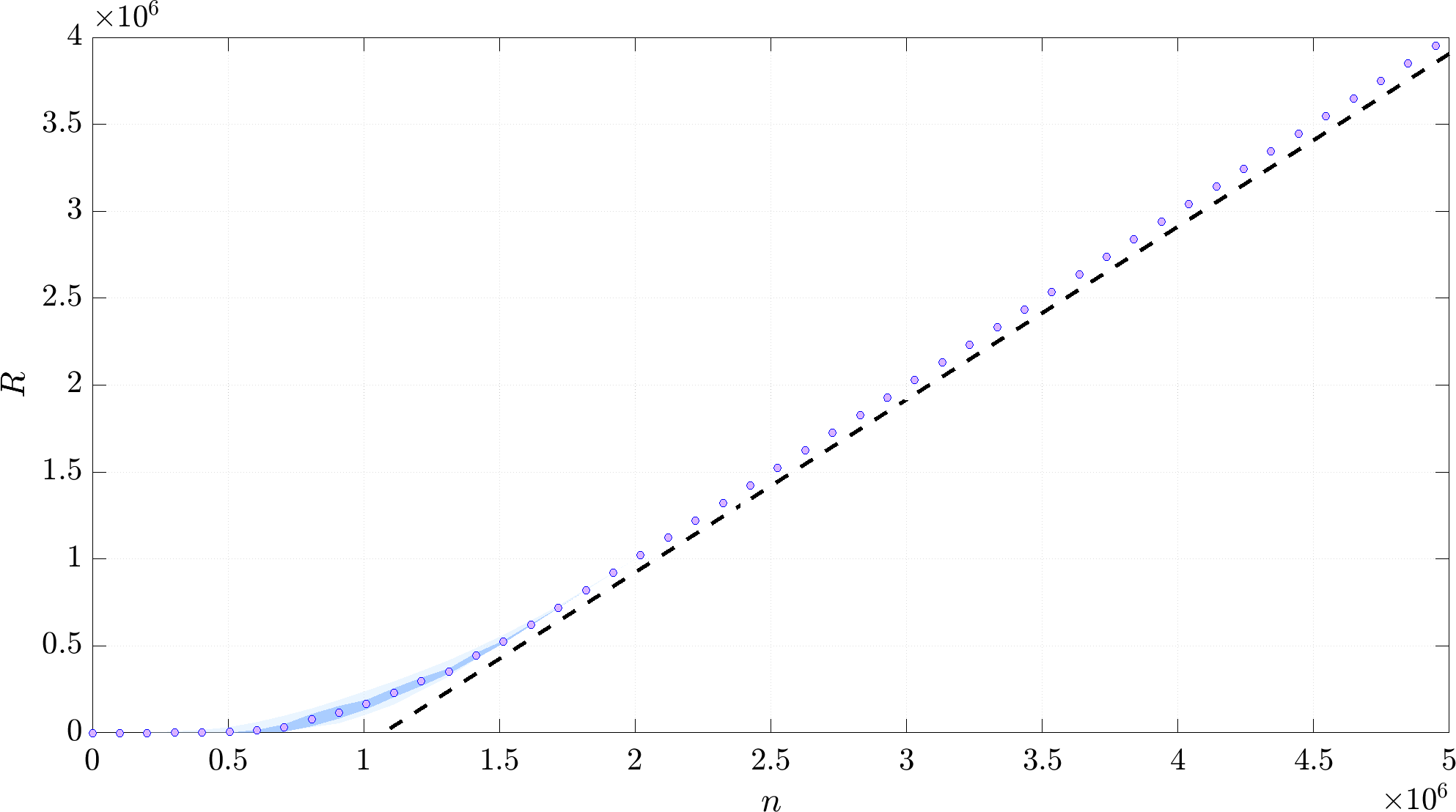}
		\caption{}
	\end{subfigure}
\caption{Total overflow, defined as the residual excess of balls assigned to destination bins beyond their capacity, computed via Monte Carlo simulations of the coarse-grained algorithm over \( S = 50 \) independent realizations. The numerical results, in blu, show perfect agreement with the closed-form analytical prediction, dashed black line, given in Eq.~\eqref{eq_overflow_C} for $n>1\cdot 10^6$, derived within the continuum-grained framework.}
\label{fig_overflow}
\end{figure}

\clearpage


\section{Discussion and Conclusion}

This paper has investigated the structural properties of origin-destination (OD) mobility networks through a formal analogy with combinatorial balls-into-bins problems. In our archetypal framework, each trip is represented as a distinguishable ball assigned from an origin to a destination (bin), generating a heterogeneous, directed mobility graph. By embedding this setting into probabilistic ensembles, we obtain a unified and tractable approach for analyzing flow distributions, congestion phenomena, and latent structural constraints.
Our main contributions arise from coordinating three complementary analytical scales. At the fine-grained level, occupancy and load statistics are derived via exact combinatorial enumeration, refined through asymptotic expansions. At the coarse-grained level, constraints are recast in cumulant form and analyzed through saddlepoint methods applied to moment-generating functions, recovering the fine-grained results in the sparse, large-network regime. At the continuum-grained level, we model interactions through latent-variable kernels and inhomogeneous random graphs, deriving macroscopic observables via integral transforms and the Laplace method. A large-deviation principle for the graphon density enforces consistency with discrete frameworks. This synthesis delivers closed-form approximations and tractable integral representations, complementing and extending our Monte Carlo simulations.

Beyond methodological advances, our framework provides actionable insights into key mobility metrics. In the occupancy problem, the degree of a location serves as a fundamental centrality measure, revealing spatial heterogeneities in connectivity and identifying hubs that dominate accessibility and flow concentration \cite{jacob2017measure}. Empirically, scale-free degree distributions, also reproduced in our data-driven simulations, highlight the dual role of hubs as efficiency drivers and vulnerability points. From a policy perspective, these distributions can inform targeted interventions: high-degree hubs may require capacity upgrades, multimodal integration, and congestion management, while low-degree areas may benefit from improved connectivity to mitigate spatial inequalities.
We have also analyzed allocation-load phenomena, vacancy, coverage, and overflow, consistently across the three scales. Vacancy quantifies underutilized locations, guiding strategies to improve accessibility in underserved areas. Coverage adopts a stopping-time perspective to inform service frequency, routing, and resource allocation for timely reach across the network. Overflow measures capacity exceedance, serving as an indicator of congestion and operational stress at critical destinations. Together, these metrics form a compact diagnostic toolkit for balancing utilization, improving equity, and strengthening resilience in urban mobility systems. Understanding these commuting dynamics can help policymakers and urban planners address inefficiencies, congestion, and mismatches in job-market accessibility.

Looking forward, our framework can be extended to incorporate additional structural constraints and network formation mechanisms. A natural step is the development of fine-grained and coarse-grained ensembles under fixed row margins, imposing constraints defined by the degree sequence, in the spirit of configurational models \cite{fosdick2018configuring,casiraghi2021configuration}\footnote{We note that coarse-grained ensemble graphs reduce to canonical ensemble graphs \cite{squartini2015breaking,cimini2019statistical,squartini2017maximum,giuffrida2023description,bianconi2009entropy} when soft constraints, i.e., constraints on expected values, are imposed.}. Further structural properties can be incorporated, such as degree homophily (assortativity), the tendency of nodes with similar degrees to connect preferentially, and degree transitivity (clustering), which captures triadic closure. Embedding these features within our ensembles would enrich the realism of the models, enabling more accurate predictions and a deeper understanding of emergent structural patterns in mobility networks.
In addition, two complementary directions appear especially promising. The first is the use of multilayer inhomogeneous random graph representations \cite{de2013mathematical,caldarelli2025lessons}  to capture the multimodal structure of transportation systems, balancing realism and tractability in modeling interactions that span multiple modes. The second is temporal network analysis, where mobility is represented as an evolving sequence of graphs \cite{hanneke2010discrete,zhang2017random,vanni2024visit}, enabling flow-based centrality measures, stochastic assignment under capacity limits, and resilience analysis of time-dependent congestion. 

Embedding these advances within ensemble formulations would yield models that remain mathematically rigorous while more closely aligned with the operational complexities of real mobility systems, where analytical methods and numerical simulations work hand in hand to capture, comprehend, and reproduce the underlying dynamics of real phenomena, while also supporting robust evaluation of policy interventions.


		\appendix
		\section{Appendix: Proofs}\label{sec:sample:appendix}

\subsection{Proof of Theorem \ref{th_macromicro}}\label{app_th_macromicro}

\begin{proof}[Coarse-grained]
	Under the asymptotic limit of infinitely many attempts, $\ell \to \infty $, and $p_{i,j}\ll 1$ so that $\ell p_{i,j} < \infty $ , we can write the conditional distribution of coarse-grained graphs as: 
	\begin{align*}
	\mathbb{P}_{\mathcal{C}}\left(\mathcal{A} \mid \sum_{i,j=1}^N a_{ij} = n\right) & = 
	\frac{\mathbb{P}_{\mathcal{C}}(\mathcal{A})}{\mathbb{P}_{\mathcal{C}}(n)}
	=\frac{\mathbb{P}_{\mathcal{C}}\left(\mathcal{A} \cap \left(\sum_{i,j} a_{ij} = n\right)\right)}{\mathbb{P}_{\mathcal{C}}\left(\sum_{i,j} a_{ij} = n\right)}
	\end{align*}
	\paragraph{1}
	The numerator can be approximated at higher orders as:
	\begin{align*}
	\mathbb{P}_{\mathcal{C}}\left(\mathcal{A} \cap \left(\sum_{i,j} a_{ij} = n\right)\right) &= \prod_{i,j=1}^N \binom{\ell}{a_{ij}} p_{ij}^{a_{ij}} (1 - p_{ij})^{\ell - a_{ij}} \\
	& \sim e^{-\ell \sum_{i,j} p_{ij} + \frac{\ell \sum_{i,j} p_{ij}^2}{2}} 
	\prod_{i,j=1}^N \frac{(\ell p_{ij})^{a_{ij}}}{a_{ij}!} 
	\cdot \left(1 + \sum_{i,j} p_{ij}^2\right).
	\end{align*}
	For large \(\ell\), the exponential can be expanded using a Taylor series:
	\[
	e^{-\frac{\ell \sum_{i,j} p_{ij}^2}{2}} \sim 1 - \frac{\ell \sum_{i,j} p_{ij}^2}{2} + \mathcal{O}\left(\left(\frac{\ell \sum_{i,j} p_{ij}^2}{2}\right)^2\right).
	\]
	With this adjustment, the full approximation becomes:
	\[
	\prod_{i,j=1}^N \binom{\ell}{a_{ij}} p_{ij}^{a_{ij}} (1 - p_{ij})^{\ell - a_{ij}} 
	\sim e^{-\ell \sum_{i,j} p_{ij}} 
	\prod_{i,j=1}^N \frac{(\ell p_{ij})^{a_{ij}}}{a_{ij}!} 
	\cdot \left(1 + \sum_{i,j} p_{ij}^2\right) 
	\cdot \left(1 - \frac{\ell \sum_{i,j} p_{ij}^2}{2}\right).
	\]
	
	To combine the higher-order corrections:
	\[
	\left(1 + \sum_{i,j} p_{ij}^2\right) \cdot \left(1 - \frac{\ell \sum_{i,j} p_{ij}^2}{2}\right) 
	\sim 1 + \sum_{i,j} p_{ij}^2 - \frac{\ell \sum_{i,j} p_{ij}^2}{2}.
	\]
	Thus, the final expression of the canonical distribution, becomes:
	\begin{align}
	\mathbb{P}_{\mathcal{C}}(\mathcal{A})
	& \sim e^{-\ell \sum_{i,j} p_{ij}} 
	\prod_{i,j=1}^N \frac{(\ell p_{ij})^{a_{ij}}}{a_{ij}!} 
	\cdot \left(1 + \sum_{i,j}\frac{a_{ij}p_{ij}}{\ell} - \frac{\ell \sum_{i,j} p_{ij}^2}{2}\right).
	\end{align}
	
	\paragraph{2} The denominator can be found evaluating the generating function as:
	$$
	\mathbb{P}_{\mathcal{C}}(n)=
	\frac{1}{n!} \frac{\partial^n}{\partial z^n} \exp \left(\ell \sum_{i,j=1}^N \ln \left( 1 + (z-1)p_{ij} \right) \right)\Bigg|_{z=0}
	$$
	where the logarithmic term captures the generating function of a sum of independent random variables.
	Then, for \( |z - 1|p_{ij} \ll 1 \), expand \( \ln(1 + (z-1)p_{ij}) \) using the Taylor series:
	\[
	\exp \left(\sum_{i,j=1}^N \ln \left( 1 + (z-1)p_{ij} \right) \right) = \exp\left(
	(z-1) \sum_{i,j=1}^N p_{ij}
	- \frac{(z-1)^2}{2} \sum_{i,j=1}^N p_{ij}^2
	+ \frac{(z-1)^3}{3} \sum_{i,j=1}^N p_{ij}^3
	- \cdots
	\right)
	\]
	\[
	\exp \left( (z-1)\langle p \rangle - \frac{(z-1)^2}{2} \langle p^2 \rangle \right) 
	\approx \sum_{k=0}^\infty \frac{((z-1)\langle p \rangle)^k}{k!} \left( 1 - \frac{(z-1)^2}{2} \langle p^2 \rangle \right).
	\]
	where we have defined:
	\[
	\langle p \rangle = \sum_{i,j=1}^N p_{ij}, \quad \langle p^2 \rangle = \sum_{i,j=1}^N p_{ij}^2.
	\]
	
	\begin{align*}
	\mathbb{P}_{\mathcal{C}}(n) \approx &  \frac{1}{n!} \frac{\partial^n}{\partial z^n} \exp \left( \ell (z-1) \langle p \rangle - \ell \frac{\langle p^2 \rangle}{2} (z-1)^2 \right) \bigg|_{z=0}
	\end{align*}

	\paragraph{3}
	Finally, in the case of large and sparse graph ($\ell \sum_{i,j=1}^Np^2_{ij}\ll 1$),  we have:
	\begin{align*}
	\mathbb{P}_{\mathcal{C}}\left(\mathcal{A} \mid \sum_{i,j=1}^N a_{ij} = n\right)&=\frac{\prod_{i,j=1}^N \frac{(\ell p_{ij})^{a_{ij}}}{a_{ij}!} e^{-\ell p_{ij}} \left[1 + \sum_{i,j}\frac{a_{ij}p_{ij}}{\ell} - \frac{\ell \sum_{i,j} p_{ij}^2}{2} \right]}
	{\frac{1}{n!} e^{-\ell \sum p_{ij}}
		(\ell \sum p_{ij})^n
		\left( 1 - \frac{\ell}{2} \frac{n(n-1)}{(\ell \sum p_{ij})^2} \sum p_{ij}^2 \right).
	} \\
	&=n! \prod_{i,j} \frac{\left(\frac{p_{ij}}{\sum p_{ij}}\right)^{a_{ij}}}{a_{ij}!}
	\left(1 + \sum_{i,j} \frac{a_{ij} p_{ij}}{\ell} - \frac{\ell}{2} \sum_{i,j} p_{ij}^2
	+ \frac{\ell n(n-1)}{2(\ell \sum p_{ij})^2} \sum p_{ij}^2 \right)
	\end{align*}
	where we have used the Stirling  approximation of Binomial coefficients in the numerator \cite{barbour1992poisson}\footnote{The expansion of the binomial coefficient for large $\ell$ so that $\ell \gg a_{ij}$ is $\binom{\ell}{a_{ij}} \sim \frac{\ell^\ell}{a_{ij}! \, \ell^{a_{ij}} (\ell - a_{ij})^{\ell - a_{ij}}}.$. Moreover for small probabilities we can write the Taylor approximation $(1 - p_{ij})^{\ell - a_{ij}} \sim e^{-\ell p_{ij}} e^{\frac{\ell p_{ij}^2}{2}} e^{a_{ij} p_{ij}}.$ 
		So that the binomial probabilities become:
		\[	\binom{\ell}{a_{ij}} p_{ij}^{a_{ij}} (1 - p_{ij})^{\ell - a_{ij}}
		\sim \frac{(\ell p_{ij})^{a_{ij}}}{a_{ij}!} e^{-\ell p_{ij}} 
		\left( 1 + \frac{a_{ij} p_{ij}}{\ell} + \frac{\ell p_{ij}^2}{2} \right).
		\] }. 
	In the denominator we have used the second order approximation of generating function as written in \ref{app_th_microcanoccu} together with property from Corollary \ref{coro_ell}, that $n \sim \ell\sum_{i,j}p_{ij}$, so the last two terms cancel out.
	
	\begin{equation}
	\mathbb{P}_{\mathcal{C}}\left(\mathcal{A} \mid \sum_{i,j=1}^N a_{ij} = n\right) = 	\frac{\mathbb{P}_{\mathcal{C}}(\mathcal{A})}{\mathbb{P}_{\mathcal{C}}(n)}= n! \prod_{i,j} \frac{\left(\frac{p_{ij}}{\sum p_{ij}}\right)^{a_{ij}}}{a_{ij}!}
	\left(1 + \sum_{i,j} \frac{a_{ij} p_{ij}}{\ell} \right)
	\end{equation}

\end{proof}

\begin{proof}[Fine-grained] 
	By using the Stirling approximation for binomial coefficients:	
	\begin{equation*}
	\binom{m_{ij}}{a_{ij}}\approx \frac{m^m e^{-m} \sqrt{2\pi m}}{a^a e^{-a} \sqrt{2\pi a} (m-a)^{m-a} e^{-(m-a)} \sqrt{2\pi (m-a)}} \approx \frac{m^{m + \frac{1}{2}}}{\sqrt{2\pi} \, a^{a + \frac{1}{2}} (m-a)^{(m-a) + \frac{1}{2}}}.
	\end{equation*}	
	we can approximate the microcanonical distribution as:
	\begin{align*}
	\mathbb{P}_m(\mathcal{A}) = &\frac{\binom{m_{11}}{a_{11}} \binom{m_{12}}{a_{12}} \cdots \binom{m_{ij}}{a_{ij}} \cdots \binom{m_{NN}}{a_{NN}}}{\binom{NM}{n}}  
	=  \frac{n! \, (NM-n)!}{(NM)!} \prod_{i,j=1}^N \frac{m_{ij}!}{a_{ij}! \, (m_{ij}-a_{ij})!}\\
	= & n! \prod_{k=0}^{n-1} \frac{1}{NM - k} \prod_{i,j=1}^N \frac{1}{a_{ij}!} \prod_{k'=0}^{a_{ij}-1} {(m_{ij} - k')}
	=  n! \prod_{i,j=1}^N \frac{1}{a_{ij}!} \left( \prod_{k=0}^{n-1} \frac{1}{NM - k} \right)^{\frac{a_{ij}}{n}} \prod_{k'=0}^{a_{ij}-1} (m_{ij} - k') \\
	=& n! \prod_{i,j=1}^N \frac{1}{a_{ij}!} \left( \frac{m_{ij}}{NM} \right)^{a_{ij}} \left( \prod_{k=0}^{n-1} \left( 1 - \frac{k}{NM} \right) \right)^{-\frac{a_{ij}}{n}} \prod_{k'=0}^{a_{ij}-1} \left( 1 - \frac{k'}{m_{ij}} \right)\\
	=& n! \prod_{i,j=1}^N \frac{1}{a_{ij}!} \left( \frac{m_{ij}}{NM} \right)^{a_{ij}} 
	e^{-\frac{a_{ij}}{n} \sum_{k=0}^{n-1} \ln\left( 1 - \frac{k}{NM} \right) + \sum_{k'=0}^{a_{ij}-1} \ln\left( 1 - \frac{k'}{m_{ij}} \right)}\\
	=& n! \prod_{i,j=1}^N \frac{1}{a_{ij}!} \left( \frac{m_{ij}}{NM} \right)^{a_{ij}} 
	e^{-\frac{a_{ij}(n-1)}{2NM} - \frac{(2n-1)a_{ij}(n-1)}{12 (NM)^2} - \frac{a_{ij}(a_{ij}-1)}{2m_{ij}} + \frac{(2a_{ij}-1)a_{ij}(a_{ij}-1)}{12 m_{ij}^2}}\\
	\approx & \; n! \exp\left(- \frac{1}{2}\sum_{i,j=1}^N \frac{a_{ij}(a_{ij}-1)}{m_{ij}} + \frac{n^2}{2NM}\right)\prod_{i,j=1}^N \frac{\left(\frac{m_{ij}}{NM}\right)^{a_{ij}}}{a_{ij}!}\\
	\end{align*}
	
	Finally, under the assumption that local sparsity scales in the same way of global sparsity than $\frac{a_{ij}}{m_{ij}}\sim \frac{n}{NM} $, in the regime of the law of large numbers, we get $\sum_{i,j=1}^N \frac{a_{ij}^2}{m_{ij}} = \frac{n^2}{NM} + \mathcal{O}\left( \frac{nN}{M} \right)
	$. So the microcanonical distribution can be written at the first order correction as:
	\begin{equation}
	P_m(\mathcal{A}) \approx n! \prod_{i,j=1}^N \frac{1}{a_{ij}!} \left( \frac{m_{ij}}{NM} \right)^{a_{ij}}\left(1+\mathcal{O}(\tfrac{nN}{M}) \right) .
	\end{equation}	
	in the asymptotic limit $nN\ll M=hN$ so $n\ll \mathfrak{n}$ 
	\footnote{In the Taylor expansion of the exponential we get that $\mathbb{P}_m(\mathcal{A})$ is:
		$$
			\approx n! \prod_{i,j=1}^N \frac{1}{a_{ij}!} \left( \frac{m_{ij}}{NM} \right)^{a_{ij}}
		\left( 1 + \tfrac{a_{ij}}{2} \left( \tfrac{n-1}{NM} - \tfrac{a_{ij}-1}{m_{ij}} \right) 
		+ \tfrac{a_{ij}^2}{8} \left( \tfrac{n-1}{NM} - \tfrac{a_{ij}-1}{m_{ij}} \right)^2 
		- \tfrac{a_{ij}}{12} \left( \tfrac{(2n-1)(n-1)}{(NM)^2} - \tfrac{(2a_{ij}-1)(a_{ij}-1)}{m_{ij}^2} \right) \right)
		$$ 
		The conditions from Taylor expansion of the exponential requires that:\[
		a_{ij}\frac{ m_{ij} n - NM a_{ij}}{NM \, m_{ij}} \ll 1 \quad \text{and} \quad 
		\frac{n^2 m_{ij}^2 - (NM)a_{ij}^2}{NM \, m_{ij} (n m_{ij} -NMa_{ij})} 
		= \frac{n m_{ij} + n \, m_{ij} a_{ij}}{NM \, m_{ij}} \ll 1.
		\]
		The sufficient conditions: the first condition implies that the total number of "allocations" $n$ is small compared to the total "capacity" $NM$, so that each link allocation behaves nearly independently because the pool (or total capacity) is much larger than the number of allocations. The local sparsity  ensures that no individual cell becomes saturated. It ensures that approximations involving binomial probabilities or multinomial coefficients remain valid. In short we have made use of sufficient conditions as:
		\[
		\frac{n}{NM} \ll 1,  \quad \frac{a_{ij}}{m_{ij}} \ll 1.
		\]
		which ensures global and  local sparsity, respectively.}.

	Finally, the goal is to determine how many total balls need to be drawn for the observed proportions   to match the expected proportions   on average or within a specific level of confidence. If the number of residents per origin location is considered infinitely large than it becomes a multinomial distribution $n! \prod_{i,j=1}^N \frac{\left(\frac{m_{ij}}{NM}\right)^{a_{ij}}}{a_{ij}!}$. 
	
	Finally, in the case of large ($ n \to \infty$) and sparse network ($n\ll NM$), we can write\footnote{The sparse graph regime is the condition where the total number of allocations $n$ is distributed sparsely across a large number of pairs $(i,j)$. In particular, we know that $a_{ij} $ represents the number of links allocated to the pair $(i,j)$ under the constraint that $\sum_{i,j=1}^N a_{ij}=n $ and $m_{ij}$ is the total possible number of links assignable to the pair $(i,j)$ and its normalized value $\pi_{ij}=\frac{m_{ij}}{NM}$ is the the relative importance (proportion) of $(i,j)$ in the network.}:
	\begin{align*}
	\mathbb{P}_{\mathcal{C}}(\mathcal{A}|n)\sim  n! \prod_{i,j=1}^N \frac{ \pi_{ij}^{a_{ij}}}{a_{ij}!} =\mathbb{P}_m(\mathcal{A}) 
	\end{align*}
	where  we have set $\pi_{ij}:=\frac{m_{ij}}{NM}=\frac{p_{ij}}{\sum p_{ij} } $.
	In particular, since $p_{ij}<1$ and $\pi_{ij}<1$ $\forall i,j$ but only $\sum_{i,j} \pi_{ij}=1$.  In this condition, in the limit of large network  $N\to \infty$ and $n\ll N$, the population size $NM$ is very large compared to the sample size $n$ so that sampling without replacement is not too much different than sampling with replacement, and hence, for  the micro-canonical ensemble graph,  the multivariate hypergeometric distribution converges to a multinomial distribution with ${\pi}_{ij}=\frac{m_{ij}}{NM}$ and applying Stirling’s formula of binomial coefficients.
\end{proof}


\subsection{Proof of Theorem \ref{th_canoccu} }\label{app_th_canoccu}
\begin{proof}
	To find the probability \( P(k_1, k_2, \dots, k_N) \), sum over all configurations \( \mathcal{A} \) that satisfy \( \sum_j a_{ij} = k_i \) for each \( i \) is for large and sparse network : 
	\begin{align*}
	P(k_1, k_2, \dots, k_N) = & \Pr\left[\bigcap_{i=1}^{N}(Y^{(M)}_{i}=k_{i} | \sum_i Y^{(M)}_{i}=n)\right] =\sum_{\substack{\mathcal{A} \\ \text{s.t. } \sum_j a_{ij} = k_i \, \forall i}} P_C(\mathcal{A} \big | \sum a_{ij}=n) \\
	=&   \frac{n!}{\prod_{i=1}^N k_i!} \prod_{i,j=1}^N \left( \frac{p_{ij}}{\sum_{j=1}^N p_{ij}} \right)^{k_i} \left[ 1 + \mathcal{O}\left( \tfrac{k_i^2}{\ell}\right) \right].
	\end{align*}
	where we have used that $(\ell p_{ij}) p_{ij} \ll 1$ and the fact that  sparse network condition includes the case that $\ell/N \ll 1$.
	
	Then we marginalize over $k_1,\ldots,k_{i-1},k_{i+1},\ldots,k_N$ such that $\sum_{i=1}^Nk_i=n$,  we can write the marginal probability density as:
	\begin{align*}
	P(k_i = k) =& \binom{n}{k} \left(\frac{\sum_j p_{ij}}{\sum_{i,j}p_{ij}} \right)^k \left( 1-\frac{\sum_j p_{ij}}{\sum_{i,j}p_{ij}} \right)^{n-k} \left[ 1 + \mathcal{O}\left( \frac{k^2}{\ell} \right) \right]
	\end{align*}
	
	Finally, the for large network  $n\gg 1$, the expected number of boxes with $k$ balls, is the occupancy:
	\begin{equation*}
	\mathbb{E}[N_k]=\sum_{i=1}^N \binom{n}{k}p_i^k (1-p_i)^k\left[ 1 + \mathcal{O}\left( \frac{k^2}{\ell} \right) \right]
	\end{equation*}
	where $p_i = \frac{\sum_j p_{ij}}{\sum_{i,j}p_{ij}}$, since $n\gg 1$ implies $\ell \gg 1$.

\end{proof}

\subsection{Proof of Theorem \ref{th_microcanoccu}}\label{app_th_microcanoccu}
\begin{proof}
	The property derives from the fact that the multivariate hypergeometric distribution is preserved when the counting variables are combined. In the microcanonical ensemble, the probability of realizing a specific graph configuration $\mathcal{A}$  with the given link distribution $\{m_{ij}\}$ is:
	\begin{equation*}
	P_m(\mathcal{A}) = \frac{\prod_{i,j} \binom{m_{ij}}{a_{ij}}}{\binom{NM}{n}},
	\end{equation*}
	where $a_{ij}$ is the number of ways to allocate $m_{ij}$ links between the pair $(i, j)$ under the constraints, i.e.  $m_{ij}$ is the number of potential links between nodes $i$ and $j$ (ball of color $j$ into the bin $i$), $a_{ij}$ is the number of allocated links between nodes  $i$ and $j$, $NM$  is the total number of possible links (balls) in the network, and $n$ is the total number of links (bins) in the network.
	Let us indicate $m_i=\sum_{j=1}^N m_{ij}$ the total number of potential links connected to node (bin) $i$ (row sums of potential links), and let $k_i=\sum_{j=1}^N a_{ij}$ to be the total number of allocated links connected to node $i$ (row sums of allocated links). To find the probability \( P(k_1, k_2, \dots, k_N) \), sum over all configurations \( \mathcal{A} \) that satisfy \( \sum_j a_{ij} = k_i \) for each \( i \):
	
	\[
	\Pr\left[\bigcap_{i=1}^{N}(Z^{(n)}_{i}=k_{i})\right] = P(k_1, k_2, \dots, k_N) = \sum_{\substack{\mathcal{A} \\ \text{s.t. } \sum_j a_{ij} = k_i \, \forall i}} P_m(\mathcal{A})
	\]
	
	Consider the numerator of \( P_m(\mathcal{A}) \):	
	\[
	\prod_{i,j} \binom{m_{ij}}{a_{ij}} = \prod_{i} \left( \prod_{j} \binom{m_{ij}}{a_{ij}} \right)
	\quad \text{and} \quad \sum_{\substack{\{ a_{ij} \} \\ \text{s.t. } \sum_j a_{ij} = k_i}} \prod_{j} \binom{m_{ij}}{a_{ij}} = \binom{m_i}{k_i}
	\]
	where we have summed over all \( \{ a_{ij} \} \) such that \( \sum_j a_{ij} = k_i \) and where \( m_i = \sum_j m_{ij} \) is the total number of potential links connected to node \( i \).
	
	The total numerator after summing over all configurations becomes \(\prod_{i} \binom{m_i}{k_i} \). Thus, the probability \( P(k_1, k_2, \dots, k_N) \) is:
	\begin{equation}\label{eq_multik}
	P(k_1, k_2, \dots, k_N) = \frac{\prod_{i} \binom{m_i}{k_i}}{\binom{NM}{n}}
	\end{equation}
	We can define $ \nu_i = {m_i/}{NM} $ as the proportion of the total number of balls  that are placed into box \( i \). and  defining the falling factorial $m_i^{(k_i)}= k_i!\binom{m_i}{k_i}$.
	\[
	\binom{m_i}{k_i} = \binom{\nu_i NM}{k_i} = \frac{(\nu_i NM)^{(k_i)}}{k_i!}
	\]
	
	where the falling factorial is:
	\[
	(\nu_i NM)^{(k_i)} = \nu_i NM \cdot (\nu_i NM - 1) \cdots (\nu_i NM - k_i + 1)=(\nu_i NM)^{k_i}\prod_{j=0}^{k_i-1} \left(1 - \frac{j}{\nu_i NM}\right)
	\]
	For the denominator we have \[
	\binom{NM}{n} = \frac{(NM)^{(n)}}{n!}=\frac{NM \cdot (NM - 1) \cdots (NM - n + 1)}{n!}=\frac{(NM)^{n}}{n!}\prod_{j=0}^{n-1} \left(1 - \frac{j}{NM}\right)
	\]
	Substituting the expressions into the probability:
	\[
	P(k_1, k_2, \dots, k_N) = \frac{n!}{\prod_{i=1}^N k_i!} \cdot \frac{\prod_{i=1}^N (\nu_i NM)^{(k_i)}}{(NM)^{(n)}}=\frac{n!}{\prod_{i=1}^N k_i!} \prod_{i=1}^N \nu_i^{k_i} \prod_{j=0}^{k_i-1} \left(1 - \frac{j}{\nu_i NM}\right) \bigg/ \prod_{j=0}^{n-1} \left(1 - \frac{j}{NM}\right)
	\]
	where \( \nu_i = \frac{m_i}{NM} \) is the share of items for node \( i \), and \( \prod_{j=0}^{k_i-1} \left(1 - \frac{j}{\nu_j NM}\right) \) accounts for finite-population corrections.
	
	To derive the marginal distribution \( P(k_i) \), we sum over all possible configurations of the remaining counts \( k_j \) for \( j \neq i \) such that \( \sum_{j \neq i} k_j = n - k_i \):
	\[
	P(k_i) = \sum_{\substack{k_1, \dots, k_{i-1}, k_{i+1}, \dots, k_N \\ \sum_{j \neq i} k_j = n - k_i}} P(k_1, k_2, \dots, k_N).
	\]
	
	Isolating the terms for node \( i \), we get  the final expression for the exact marginal probability distribution by integrating over all possible configurations of the other counts, incorporating finite-population corrections due to sampling without replacement. Namely :
	\[
	P(k_i) = \frac{n! \, \nu_i^{k_i}}{k_i!} \prod_{j=0}^{k_i-1} \left(1 - \frac{j}{\nu_i NM}\right) \cdot \frac{1}{\prod_{j=0}^{n-1} \left(1 - \frac{j}{NM}\right)} \sum_{\substack{k_1, \dots, k_{i-1}, k_{i+1}, \dots, k_N \\ \sum_{j \neq i} k_j = n - k_i}} \prod_{j \neq i} \frac{\nu_j^{k_j}}{k_j!} \prod_{l=0}^{k_j - 1} \left(1 - \frac{l}{\nu_j NM}\right).
	\]
	
	In this expression, the term \( \frac{n! \, \nu_i^{k_i}}{k_i!} \prod_{j=0}^{k_i-1} \left(1 - \frac{j}{\nu_i NM}\right) \) represents the probability of selecting exactly \( k_i \) items for node \( i \), accounting for finite-population adjustments, while the sum over the remaining \( k_j \) for \( j \neq i \) captures all configurations of the remaining counts constrained by \( \sum_{j \neq i} k_j = n - k_i \).
	The term $\sum \prod_{j \neq i} \frac{\nu_j^{k_j}}{k_j!}$  represents the exact probability of distributing $n-k_i$ balls across $N-1$ boxes, accounting for each possible configuration with finite-population effects (where the probability of each allocation depends on the previous allocations). 
	Thus we can use the  multinomial approximation for the probability of distributing $n-k_i$ balls across the remaining boxes, with an explicit Big-O error term that accounts for the finite-population corrections:
	\[
	\sum_{\substack{k_1, \dots, k_{i-1}, k_{i+1}, \dots, k_N \\ \sum_{j \neq i} k_j = n - k_i}} \prod_{j \neq i} \frac{\nu_j^{k_j}}{k_j!} \approx \frac{(1 - \nu_i)^{n - k_i}}{(n - k_i)!} + O\left(\frac{(n - k_i)(n - k_i - 1)}{N M}\right).
	\]
	
	To determine the expected number of boxes containing exactly \( k \) balls, define \( X_i^{(k)} \) as an indicator variable for each box \( i \), where \( X_i^{(k)} = 1 \) if box \( i \) has exactly \( k \) balls and \( X_i^{(k)} = 0 \) otherwise. The total number of boxes with exactly \( k \) balls is \( X^{(k)} = \sum_{i=1}^N X_i^{(k)} \). To find the expected number of boxes with exactly \( k \) balls, calculate \( \mathbb{E}[X^{(k)}] = \sum_{i=1}^N \mathbb{E}[X_i^{(k)}] \). Since \( X_i^{(k)} \) is an indicator variable, \( \mathbb{E}[X_i^{(k)}] \) is simply the probability that box \( i \) has exactly \( k \) balls, which is \( P(k_i = k) \). From the marginal distribution derived earlier, the expected number of boxes containing exactly \( k \) balls is then:
	
	\[
	\mathbb{E}[X^{(k)}] = \sum_{i=1}^N P(k_i = k).
	\]
	
	Substituting the expression for \( P(k_i = k) \) from above, in the limit of large network, we get:
	\[
	\mathbb{E}[N_k] = \sum_{i=1}^N \frac{n! \, \nu_i^k}{k!} \frac{(1 - \nu_i)^{n - k}}{(n - k)!}\left( 1+ O\left(\tfrac{n^2 + k^2}{N M}\right)\right).
	\]
	This expression represents the expected number of boxes with exactly \( k \) balls, incorporating the finite-population corrections from the hypergeometric-like structure
	\footnote{ Let us notice that we could  recover the same approximation for the expected occupancy, more easly from eq.\eqref{eq_multik}, considering that the probability that node \( i \) has exactly \( k \) items in this finite population (hypergeometric distribution) is:
		\(
		P(Z_i=k) = \frac{\binom{m_i}{k} \binom{NM - m_i}{n - k}}{\binom{NM}{n}}
		\).
		For large \( N \), we approximate each binomial coefficient:
		\begin{align*}
		\binom{m_i}{k} &\approx \frac{(\nu_i NM)^k}{k!} \left(1 + O\left(\frac{k^2}{NM}\right)\right), \\
		\binom{NM - m_i}{n - k} &\approx \frac{((1 - \nu_i) NM)^{n - k}}{(n - k)!} \left(1 + O\left(\frac{(n - k)^2}{NM}\right)\right), \\
		\binom{NM}{n} &\approx \frac{(NM)^n}{n!} \left(1 + O\left(\frac{n^2}{NM}\right)\right).
		\end{align*}
		
		Substituting these approximations, we obtain:
		\[
		P(Z_i=k) = \frac{n!}{k! (n - k)!} \nu_i^k (1 - \nu_i)^{n - k} \left(1 + O\left(\frac{k^2+ n^2}{NM}\right)\right)
		\]
		
		The expected number of nodes with exactly \( k \) items is
		$
		\mathbb{E}[N_{k}] = \sum_{i=1}^N P(Z_i=k)
		$}.
	This result provides a precise approximation of the expected count with an explicit big-O error bound, indicating the accuracy of the approximation as \( N \) grows.
\end{proof}

\subsection{Proof of Remark~\ref{coro_asym_fin}} \label{app_asym_fine}
\begin{proof}
	Let $T(n, N)$ denote the number of ways in which $n$ balls are distributed among $N$ boxes. Denote by $Q(n, m, r, k)$ the number of ways of distributing the $n$ balls among the $m=N$ boxes so that exactly $r$ boxes have exactly $k$ balls. Then if we suppose that each way is equally likely $\Pr[M_j=r]=Q(n, m, r, k)/T(n, m)$. However, when considering boxes to have different occupancy proportions $\pi_i$, we follow the approach of generating function of occupancy distribution as described in \cite{johnson1977urn}[Ch.3]. We can divide our collection into two groups containing $m'$ and $m''=m-m'$ boxes, we can write:
	\[
	\Pr[M_j = m_j \mid n;m] = \sum_{\substack{n'+ n'' = n,\\g'+ g{''}=g}} \binom{n}{n'} \left( \sum_{i=1}^{m'} \pi_i \right)^{n'} \left( \sum_{i=1}^{m''} \pi_{m'+i} \right)^{n''} \Pr[M_j' = g \mid n'; m'] \Pr[M_j'' = g'' \mid n'', m'']
	\]
	At this point, by using the generating functions of occupancy distributions:
	\begin{align*}
	H_j^{(m)}(z, x) =&  \sum_{n,g=0}^{\infty} \frac{z^{n}}{n!} x^{g} \Pr[M_j = g \mid n,m]
	= H_j^{(m')}\left(z \sum_{i=1}^{m'} \pi_i, x \right) H_j^{(m'')}\left(z \sum_{i=1}^{m''} \pi_{m'+i}, x \right)
	= \prod_{i=1}^{m} H_j^{(1)}(z \pi_i, x)
	\end{align*}
	where we have defined:
	\[
	H_j^{(1)}(z \pi_i, x) = \sum_{n,g=0}^{\infty} \frac{(z \pi_i)^n}{n!} x^{g} \Pr[M_j = g \mid n,1]
	\]
	Where there is only one box, then necessarily it contains all the available balls, and so:
	\[
	\Pr[M_j = g \mid n; 1] = 
	\begin{cases} 
	1 & \text{if } g = 1, n = j, \text{ or } g = 0, n \neq j, \\
	0 & \text{otherwise}.
	\end{cases}
	\]
	So:
	\begin{align*}
	H_j^{(m)}(z \pi_i, x) =& \prod_{i=1}^m e^{z\pi_i} +{(x-1)}\frac{(z\pi_i)^j}{j!}  
	\end{align*}
	Finally, the expected value \( \mathbb{E}[M_j] \)  is calculated by taking the partial derivative of \( H_j^{(m)}(z, x) \) with respect to \( x \), followed by the partial derivative with respect to \( z \), and then evaluating at \( x=1 \) and \( z=0 \):
	\begin{align*}
	\mathbb{E}[N_k] = & \left. \frac{\partial^n}{\partial z^n} \frac{\partial}{\partial x} H_k^{(m)}(z, x) \right|_{\substack{x=1,\\ z=0}}
	= \frac{\partial^n}{\partial z^n} \sum_{i=1}^{m} (z P_i) e^{z (1-\pi_i)} \Big|_{\substack{z=0}}  = \frac{\partial}{\partial z} \sum_{i=1}^{m} \sum_{n'=0}^{\infty} \frac{z^{n'+k} \pi_i^{j}}{n'!k!} (1-\pi_i)^{n'} \Big|_{z=0}\\
	=& \sum_{i=1}^{m} \binom{n}{k}\pi_i^k (1 - \pi_i)^{n-k}
	\end{align*}
	
\end{proof}

\subsection{Proof of Proposition \ref{prop_cvc}}\label{app_cvc}
\begin{proof}
	The proof derives directly from the results proved and discussed in \cite{sevast1973poisson,holst1977some} where we have taken into account that  the parameter \( m \) is the {limiting expected number of empty bins}, formally defined by the asymptotic condition:
	$
	m := \lim_{N \to \infty} \sum_{i=1}^N (1 - \nu_i)^n.
	$
	The asymptotic regime of large and sparse graphs implies that:
	\begin{equation}
	\lim_{N \to \infty} \max_{1 \leq i \leq N} (1 - \nu_i)^n = 0 \quad , \quad \lim_{N \to \infty} \sum_{i=1}^N (1 - \nu_i)^n < \infty.
	\end{equation}	
	For large \( n \) and small \( \nu_i \), use the approximation:
	\[
	(1 - \nu_i)^n \approx e^{-n \nu_i} \quad \Rightarrow \quad \hat{m} \approx \sum_{i=1}^N e^{-n \nu_i}
	\]	\qed
	
	Let $T_b$ is the first time $n$ such that exactly $b$ bins remain empty,	i.e. it is  the stopping time  \( T_b=\min\{n:E_n=b \} \)  to be  the (random) number of balls required until exactly \( b \) bins remain empty.  In the sparse, balanced regime where $\nu_k \sim \tfrac{1}{N}$, we approximate:
	\[
	f_N(T_b) := \sum_{k=1}^N (1 - \nu_k)^{T_b} \approx \sum_{k=1}^N e^{-\nu_k T_b}.
	\]
	The balanced regime condition is also referred to as a regularity or uniform boundedness assumption and it ensures that every bin has a comparable chance of being selected, maintaining a stable, analyzable structure.
	The vacancy problem investigates $T_b$ i.e. the throw on which for the first time exactly $b$ cells remain empty. This is a hitting-time or first passage time problem.  
	Now, let $Z_N := \sum_{k=1}^N e^{-\nu_k T_b}$, then, for fixed $b$, the number of empty bins after $n$ throws is approximately:
	$
	E_n \sim \text{Po}\left( \sum_{k=1}^N e^{-n \nu_k} \right).
	$
	
	Set $n = T_b$ such that $E_n = b$. Conditioning on this Poisson value, the Poisson parameter $Z_N$ has the distribution:
	\[
	Z_N=f_N(T_b) \xrightarrow{d} \text{Gamma}(b+1, 1) = \tfrac{1}{2} \chi^2_{2(b+1)}.
	\]
	which indicates the expected number of trips still empty in a new trip observation of length  $T_b$, and, for $b\gg 1$.
	To estimate \( \mathbb{E}[T_b] \), we use a delta-method expansion around \( n_b \), noting that \( \mathbb{E}[E_{T_b}] = b \), so the expansion yields
	\[
	\mathbb{E}[T_b] = n_b + \frac{\operatorname{Var}[E_{n_b}]}{f_N'(n_b)} + o(1),
	\quad \text{where } f_N(t) = \sum_{i=1}^N e^{-t \nu_i}.
	\]
	
	In the balanced case where \( \nu_i \sim 1/N \), we estimate:
	\[
	f_N'(n_b) = \sum_{i=1}^N \nu_i e^{-n_b \nu_i} = \Theta\left(\frac{b}{N}\right),
	\quad
	\operatorname{Var}[E_{n_b}] = \sum_{i=1}^N e^{-n_b \nu_i} \left(1 - e^{-n_b \nu_i} \right) = \Theta(b).
	\]
	
	Substituting into the expansion gives:
	$
	\frac{\operatorname{Var}[E_{n_b}]}{f_N'(n_b)} = \Theta(N),
	\quad \Rightarrow \quad
	\mathbb{E}[T_b] = n_b + \Theta(N),
	$
	where \( n_b \) is the solution of
	$
	\sum_{i=1}^N e^{-n_b \nu_i} = b.
	$
	
	We can, so, write 
	\begin{equation*}
	\mathbb{E}[T_b] = n_b + \Theta(N),
	\quad \text{where } n_b \text{ solves } \sum_{i=1}^N e^{-n_b \nu_i} = b.
	\end{equation*}
	Under balance condition of $0<C<N\nu_i<D<+\infty$, we can approximate using the Laplace method applied as a boundary approximation:
	\begin{equation*}
	n_b\approx \frac{1}{\min_i \nu_i}\log \frac{N}{b}\approx \frac{N}{C}\log\frac{N}{b}
	\end{equation*}

\end{proof}

\section{Appendix: Graph Code via Stub-Matching}\label{app_stub}
An effective alternative to fine-grained graph generation is the stub-matching approach, inspired by \cite{bayati2010sequential,van2021sequential}, which can be extended to incorporate latent variables. Here, the number of slots between each ordered pair $(i,j)$ is determined by a capacity matrix $m = \{m_{ij}\}$ generated from a multinomial distribution whose parameters are derived from node-specific latent variables. Each node $i$ is assigned a latent attractiveness $x_i$ and a latent trip-generation potential $y_i$, drawn respectively from a truncated Pareto and an exponential distribution. These latent variables determine the edge propensities $p_{ij} = x_i^\alpha y_j^\beta$, which are normalized to define the probabilities $\pi_{ij} = p_{ij} / \sum_{u,v} p_{uv}$. The matrix $m$ is then drawn from $\mathrm{Multinomial}(NM, \{\pi_{ij}\})$, assigning each node pair a number of available slots summing to $NM$ as shown in the pseudo-code \ref{code_stub}.

A specific graph realization is constructed by creating a pool $\mathcal{S}$ containing $m_{ij}$ stubs labeled by $(i,j)$ for all node pairs. Exactly $n$ stubs are then selected uniformly at random (without replacement) from this pool, with each selected stub incrementing the corresponding entry $a_{ij}$ in the adjacency matrix $A = \{a_{ij}\}$. This procedure ensures that $0 \leq a_{ij} \leq m_{ij}$ and $\sum_{i,j} a_{ij} = n$.  

\begin{algorithm}[H] \label{code_stub}
	\caption{Sampling of a fine-grained Latent Variable Graph via Stub-Matching}
	\SetKwInOut{Input}{Input}\SetKwInOut{Output}{Output}
	\Input{Number of nodes $N$; total number of links $n$; total number of slots $NM$}
	\Output{Adjacency matrix $A = \{a_{ij}\}$ such that $a_{ij} \le m_{ij}$ and $\sum a_{ij} = n$}
	
	\BlankLine
	\textbf{Step 1: Sample latent variables}\;
	\For{$i \gets 1$ \KwTo $N$}{
		$x_i \gets$ draw from Pareto$(x_{\min}=1, \eta=2)$\;
		$y_i \gets$ draw from Exponential$(\lambda=1/5)$\;
	}
	
	\BlankLine
	\textbf{Step 2: Compute pairwise propensities and normalize}\;
	\For{$i,j \in \{1,\ldots,N\}$}{
		$p_{ij} \gets x_i^\alpha \cdot y_j^\beta$\;
	}
	$\pi_{ij} \gets p_{ij} / \sum_{u,v} p_{uv}$\;
	
	\BlankLine
	\textbf{Step 3: Generate capacity matrix from multinomial distribution}\;
	$(m_{11}, \ldots, m_{NN}) \gets \text{Multinomial}(NM, \{\pi_{ij}\})$\;
	
	\BlankLine
	\textbf{Step 4: Create stub pool of total size $NM$}\;
	Initialize list $\mathcal{S}$ as an array with $m_{ij}$ copies of each pair $(i,j)$\;
	
	\BlankLine
	\textbf{Step 5: Allocate $n$ links by drawing from the stub pool without replacement}\;
	Select $n$ elements $(i,j)$ from $\mathcal{S}$ uniformly at random (without replacement)\;
	\For{each sampled pair $(i,j)$}{
		$a_{ij} \gets a_{ij} + 1$\;
	}
	
	\BlankLine
	\Return{$A = \{a_{ij}\}$}
\end{algorithm}

Let us observe that in the sparse regime ($n \ll NM$ and $a_{ij} \ll m_{ij}$ for all $i,j$), the distribution of $m_{ij}$ is sharply peaked around its mean $\mathbb{E}[m_{ij}] = NM \pi_{ij}$ with fluctuations of order $\sqrt{NM \pi_{ij}(1-\pi_{ij})}$, and the probability of exceeding slot capacities is negligible. In this limit, the law of large numbers ensures that the conditional law for $a_{ij}$ approximated by a binomial or Poisson distribution:
\[
\mathbb{P}(a_{ij}\,|\,m_{ij}) \approx \binom{m_{ij}}{a_{ij}} \left(\frac{n}{NM}\right)^{a_{ij}} \left(1-\frac{n}{NM}\right)^{m_{ij}-a_{ij}}.
\]

As a result, the stub-matching ensemble provides a computationally efficient and analytically tractable approximation.

		\bibliographystyle{elsarticle-num} 
		\bibliography{cas-refs}
		
		
		
		
		
	\end{document}